\documentclass{article}

\usepackage{graphicx}
\usepackage{multirow}
\usepackage{amssymb} 
\usepackage{amsthm}
\usepackage{amsmath}
\usepackage{fullpage}
\usepackage[longnamesfirst]{natbib}
\usepackage{enumerate}
\usepackage{amsfonts, bm, amsmath, color, natbib, rotating, amsthm, subfigure, lscape, multicol, epstopdf, verbatim, hyperref}  
\usepackage{authblk}
\usepackage{xr}
\bibpunct{(}{)}{;}{a}{,}{,}

\usepackage{setspace}
\theoremstyle{plain} \newtheorem{theorem}{Theorem} \newtheorem{proposition}{Proposition} \newtheorem{lemma}{Lemma} 

\theoremstyle{definition}    
\newcommand{\utwi}[1]{\mbox{\boldmath $ #1$}}
\theoremstyle{remark}   
\begin{document}

\newif\ifblinded

\title{Dynamic Function-on-Scalars Regression}

\ifblinded
\author{}
\else

\author{Daniel R. Kowal\thanks{Assistant Professor, Department of Statistics, Rice University, Houston, TX 77251-1892 (E-mail: \href{mailto:daniel.kowal@rice.edu}{daniel.kowal@rice.edu}).}}

\fi

\maketitle
\large 

\vspace{-10mm}

\begin{abstract}
We develop a modeling framework for dynamic function-on-scalars regression, in which a time series of functional data is regressed on a time series of scalar predictors. The regression coefficient function for each predictor is allowed to be dynamic, which is essential for applications where the association between predictors and a (functional) response is time-varying. For greater modeling flexibility, we design a nonparametric reduced-rank functional data model with an unknown functional basis expansion, which is data-adaptive and, unlike most existing methods, modeled as unknown for appropriate uncertainty quantification. Within a Bayesian framework, we introduce shrinkage priors that simultaneously (i) regularize time-varying regression coefficient functions to be locally static, (ii) effectively remove unimportant predictor variables from the model, and (iii) reduce sensitivity to the dimension of the functional basis. A simulation analysis confirms the importance of these shrinkage priors, with notable improvements over existing alternatives. We develop a novel projection-based Gibbs sampling algorithm, which offers unrivaled computational scalability for fully Bayesian functional regression. We apply the proposed methodology (i) to analyze the time-varying impact of macroeconomic variables on the U.S. yield curve and (ii) to characterize the effects of socioeconomic and demographic predictors on age-specific fertility rates in South and Southeast Asia.
\end{abstract}
\clearpage

\section{Introduction}
We are interested in modeling the association between a \emph{functional} response and \emph{scalar} predictors, commonly referred to as \emph{function-on-scalars regression} (FOSR); see \cite{silverman2005functional} and \cite{morris2015functional}. We address the additional complication that the functional response and the scalar predictors are both \emph{time-ordered}. Applications of time-ordered functional data, or \emph{functional time series}, are abundant, including: daily interest rate curves as a function of time to maturity \citep{FDFM,kowal2017functional}; yearly sea surface temperature as a function of time-of-year \citep{besse2000autoregressive}; yearly mortality rates as a function of age \citep{hyndman2007robust}; daily pollution curves as a function of time-of-day \citep{damon1z2002inclusion,aue2015prediction}; and a collection of spatio-temporal applications in which a time-dependent variable is measured as a continuous function of spatial location (e.g., \citealp{cressie2011statistics}). In these applications and others, there may be interest in modeling the relationship between the functional time series and dynamic predictors.

In functional regression, a fundamental challenge is appropriately accounting for within-curve dependence, or smoothness, while simultaneously modeling the effects of predictor variables. In the dynamic setting, the time-ordering of functional data and predictors introduces further complications. Unmodeled (time) dependence produces statistically inefficient estimators and can lead to incorrect inference and spurious relationships. In many applications, the association between predictors and the functional response may be time-varying. \cite{dangl2012predictive} discuss the importance of time-varying parameter regression for macroeconomic data, but the concepts are  broadly applicable: structural shifts obscure (dynamic) relationships and produce inferior estimates, predictions, and forecasts. It is therefore essential to account for both \emph{time-dependence} and \emph{time-variation}. 


We propose a Bayesian \emph{dynamic function-on-scalars regression} (DFOSR) model to jointly model \emph{within-curve} (functional) dependence, \emph{between-curve} (time) dependence, and \emph{dynamic associations} with scalar predictors. Within-curve dependence is modeled nonparametrically using a reduced-rank functional data model, which provides model flexibility for broad applicability. The unknown basis functions are endowed with a prior distribution that encourages smoothness, produces data-adaptive basis functions, and incorporates uncertainty quantification via the posterior distribution. We introduce an autoregressive structure for between-curve dependence and model the dynamic predictors by extending \emph{time-varying parameter regression} to the functional data setting. Time-varying parameter regression has successfully improved estimation and forecasting for scalar time series \citep{dangl2012predictive,korobilis2013hierarchical,belmonte2014hierarchical,kowal2017dynamic}, but to the best of our knowledge has not yet been used for functional data. We introduce shrinkage priors that simultaneously guard against overfitting yet 
preserve model flexibility.  A simulation study (Section \ref{simulations}) confirms the importance of these priors and demonstrates decisive improvements in statistical efficiency and uncertainty quantification relative to existing alternatives. Computationally scalable posterior inference is achieved using an efficient Gibbs sampling algorithm. The model is applicable for both densely- and sparsely-observed functional data (see Sections \ref{yields} and \ref{fertility}, respectively), with a model-based imputation procedure for the latter case. 

Our methodology is motivated by two applications. First, we study the impact of macroeconomic variables on the U.S. yield curve. For a given currency and level of risk of a debt, the yield curve describes the interest rate at a given time as a \emph{function} of the length of the borrowing period, or time to maturity, and evolves over \emph{time}. We study the dynamic associations between U.S. interest rates and several fundamental components in the U.S. economy, in particular real activity, monetary policy, and inflation. Building upon the setting in \cite{diebold2006macroeconomy}, our approach (i) relaxes the parametric (Nelson-Siegel) assumption for the functional component, (ii) allows for the macroeconomic associations with the yield curve to be time-varying, (iii) incorporates a model for volatility clustering, and (iv) provides fully Bayesian inference and joint estimation of model parameters. As a result, we gain insight into how these important macroeconomic variables are related to interest rates of different maturities, and how these relationships vary over time.

Second, we analyze \emph{age-specific fertility rates} (ASFRs) for developing nations in South and Southeast Asia. ASFRs measure fertility as a \emph{function} of age within a population, which changes over \emph{time}, and may depend on socioeconomic and demographic \emph{predictor variables}. Fertility is a fundamental component in population growth, with major implications for planning and allocation of resources. Our methodology provides a mechanism for understanding how various socioeconomic and demographic variables impact the \emph{shape} of the ASFR, which allows for differential age-specific effects with appropriate uncertainty quantification.

The remainder of the paper is organized as follows: we introduce the model in Section \ref{dfosr}; the reduced-rank functional data model is  in Section  \ref{loadings}; the shrinkage priors are  in Section \ref{shrinkage}; a simulation analysis is in Section \ref{simulations}; we apply the model to yield curves in Section \ref{yields} and age-specific fertility rates in Section \ref{fertility}; the MCMC algorithm is in Section \ref{MCMC}; and we conclude in Section \ref{conclusions}. Supplementary files include an \texttt{R} package available on Github, the yield curve and fertility datasets, and an Appendix with additional details on the MCMC algorithm, simulations, and the applications.

\section{Dynamic Function-on-Scalars Regression}\label{dfosr}
Let $\{Y_t(\bm\tau)\}_{t=1}^T$ be a time-ordered sequence of functional data with $\bm\tau \in \mathcal{T}$, where $\mathcal{T} \subset \mathbb{R}^D$ is a compact index set and $D \in \mathbb{Z}^+$. Suppose we have time-ordered predictors $\bm x_t = (x_{1,t},\ldots, x_{p,t})'$ and we are interested in modeling the association between the scalar predictors $ x_{j,t}$ and the functional response $Y_t$. We consider the setting in which the relationship between $x_{j,t}$ and $Y_t$ may be time-varying. The proposed \emph{dynamic function-on-scalars regression} (DFOSR) model has three levels, which are jointly expressed via \eqref{fts}-\eqref{evol} below. 

First, we decompose the functional time series $Y_t$ into a linear combination of $K$  \emph{loading curves}, $\{f_k(\bm\tau)\}_{k=1}^K$, and  \emph{factors}, $\{\beta_{k,t}\}_{k=1}^K$, for each time $t=1,\ldots,T$:
\begin{equation}
\label{fts}
Y_t(\bm \tau) = \sum_{k=1}^K f_k(\bm \tau) \beta_{k,t} + \epsilon_t(\bm \tau), \quad \epsilon_t(\bm\tau) \stackrel{indep}{\sim}N(0, \sigma_{\epsilon_t}^2), \quad \bm \tau \in \mathcal{T}
\end{equation}
Model \eqref{fts} is a \emph{dynamic functional factor model}: the loadings $\{f_k\}$ are modeled as smooth unknown functions of $\bm \tau$ to account for the within-curve correlation structure in $Y_t$, and the factors $\{\beta_{k,t}\}$ are modeled dynamically to account for the between-curve time dependence in $Y_t$. Equivalently, we may interpret $\{f_k\}$ as a time-invariant functional basis for $Y_t$ with dynamic basis coefficients  $\{\beta_{k,t}\}$, which we model using dynamic predictor variables (see  \eqref{reg} below). 
Each $f_k$ is modeled nonparametrically using \emph{low-rank thin plate splines}, which are well-defined for $\mathcal{T}\subset \mathbb{R}^D$ with $D \in \mathbb{Z}^+$ and are smooth, flexible, and efficient to compute \citep{ruppert2003semiparametric,wood2006generalized}. By modeling the $\{f_k\}$ as unknown, and imposing suitable identifiability constraints (see Section \ref{orthog}), our model incorporates the uncertainty of $\{f_k\}$ into the posterior distribution for all parameters of interest, which is necessary for valid inference. Model \eqref{fts} assumes conditionally Gaussian errors $\epsilon_t(\cdot)$, possibly with dynamic variance $\sigma_{\epsilon_t}^2$ to account for volatility clustering (see Section \ref{yields}). 

Next, we introduce a dynamic regression component to incorporate the predictors $x_{j,t}$:
\begin{equation}
\label{reg}
\beta_{k,t} = \mu_{k} +  \sum_{j=1}^p x_{j,t} \alpha_{j,k,t} + \gamma_{k,t}, \quad \gamma_{k,t} = \phi_k \gamma_{k,t-1} + \eta_{k,t}, \quad \eta_{k,t} \stackrel{indep}{\sim}N(0, \sigma_{\eta_{k,t}}^2) 
\end{equation}
where $\mu_k$ is the intercept for factor $k$, $\alpha_{j,k,t}$ is the time-varying regression coefficient for predictor $j$ and factor $k$ at time $t$, and $\gamma_{k,t}$ is the regression error term, which we allow to be autocorrelated via an AR(1) process. Extensions to more general time series models for $\gamma_{k,t}$ in \eqref{reg}, such as ARIMA models, may be easily incorporated into the proposed model framework. Each regression coefficient $ \alpha_{j,k,t}$ varies with $k$, and therefore its association with $Y_t(\bm \tau)$ for a particular $\bm \tau$ may be interpreted via the loading curve $f_k(\bm \tau)$.

Lastly, we specify the dynamics---and regularization---for the regression coefficients, $\alpha_{j,k,t}$:
\begin{equation}
\label{evol}
\alpha_{j,k,t} = \alpha_{j,k,t-1}  + \omega_{j,k,t}, \quad \omega_{j,k,t} \stackrel{indep}{\sim}N(0, \sigma_{\omega_{j,k,t}}^2)
\end{equation}
 For each $k$, \eqref{reg}-\eqref{evol} is a time-varying parameter regression for the dynamic predictors $x_{j,t}$, where the factors $\beta_{k,t}$ operate as the response variable. We select priors for $\sigma_{\omega_{j,k,t}}^2$ in Section \ref{shrinkage} to encourage shrinkage of $\alpha_{j,k,t}$. Locally, we shrink $\omega_{j,k,t}$ toward zero, which implies that $\alpha_{j,k,t} \approx \alpha_{j,k,t-1}$ is locally constant at time $t$. Importantly, the factor-specific regression coefficients $\alpha_{j,k,t}$ are allowed to change at any time $t$, which may capture structural shifts, but the shrinkage prior encourages a more parsimonious model. Globally, we shrink $\omega_{j,k,t}$ toward zero for all $t$, which, combined with shrinkage of the initial state $\alpha_{j,k,0}$, effectively removes factor $k$ for predictor $j$ from the model. Finally, we introduce \emph{ordered} shrinkage across $k=1,\ldots,K$  to cumulatively reduce the relative importance of the higher number factors $k$, which mitigates the impact of the choice of $K$, as long as $K$ is chosen sufficiently large. The simulation analysis in Section \ref{simulations} validates the importance of these shrinkage priors.  

The DFOSR  \eqref{fts}-\eqref{evol} also induces a model representation in the functional  $\bm \tau \in \mathcal{T}$ space. Let $\mathcal{GP}(c, C)$ denote a Gaussian process with mean function $c$ and covariance function $C$. 
\begin{proposition} \label{model_equiv}
Model \eqref{fts}-\eqref{evol} implies the dynamic functional regression model
\begin{align}
\label{dfosr_fts}
Y_t(\bm \tau) &=  \tilde \mu(\bm \tau) + \sum_{j=1}^p x_{j,t}  {\tilde \alpha}_{j,t}(\bm \tau) + \tilde \gamma_t(\bm \tau) + \epsilon_t(\bm \tau), \quad \epsilon_t(\bm\tau) \stackrel{indep}{\sim}N(0, \sigma_{\epsilon_t}^2), \quad \bm \tau \in \mathcal{T} \\
\label{dfosr_reg}
\tilde \gamma_t(\bm \tau) &= \int \tilde \phi(\bm \tau, \bm u)  \tilde \gamma_{t-1}(\bm u) \, d\bm u + \tilde \eta_t(\bm \tau), \quad  \tilde \eta_t(\cdot) \stackrel{indep}{\sim} \mathcal{GP}(0, C_{\eta_t}) \\
\label{dfosr_evol}
\tilde \alpha_{j,t}(\bm \tau) &=\tilde \alpha_{j,t-1}(\bm \tau) + \tilde \omega_t(\bm \tau), \quad  \tilde \omega_{j,t}(\cdot) \stackrel{indep}{\sim} \mathcal{GP}(0, C_{\omega_{j,t}}) 
\end{align}
under the expansions $\tilde \mu(\bm \tau) = \sum_k f_k(\bm\tau) \mu_k$, $ {\tilde \alpha}_{j,t}(\bm \tau)  = \sum_k f_k(\bm \tau)  \alpha_{j,k,t}$, $ \tilde \gamma_t(\bm \tau) = \sum_k f_k(\bm \tau) \gamma_{k,t}$, $\tilde \phi(\bm \tau, \bm u) = \sum_{k} f_k(\bm \tau) f_k(\bm u) \phi_k$, $\tilde \eta_t(\bm \tau) = \sum_k f_k(\bm \tau) \eta_{k,t}$, $\tilde \omega_{j,t}(\bm \tau) = \sum_k f_k(\bm \tau) \omega_{j,k,t}$, and the covariance functions $C_{\eta_t}(\bm \tau, \bm u) = \sum_k f_k(\bm \tau) f_k(\bm u) \sigma_{\eta_{k,t}}^2$ and $C_{\omega_{j,t}}(\bm \tau, \bm u)= \sum_k f_k(\bm \tau) f_k(\bm u) \sigma_{\omega_{j,k,t}}^2$. 
\end{proposition}
The predictors $x_{j,t}$ are directly associated with the functional time series $Y_t(\bm \tau)$ via the dynamic regression coefficient functions $\tilde \alpha_{j,t}(\bm \tau)= \sum_k f_k(\bm \tau)  \alpha_{j,k,t}$. Since we obtain MCMC draws from the posterior distribution of $\{f_k\}$ and $\{\alpha_{j,k,t}\}$, we may conduct posterior inference on $\tilde \alpha_{j,t}(\bm \tau)$ directly without modifying the MCMC sampling algorithm. The error term $\tilde \gamma_t(\bm \tau)$ captures the large-scale variability in $Y_t(\bm \tau)$ at time $t$, and is autocorrelated, while the error term $\epsilon_t(\bm \tau)$ models the small-scale variability, i.e., the observation error. Equation \eqref{dfosr_reg} is a \emph{functional autoregressive model} for $\tilde \gamma_t(\bm \tau)$, which is the functional data analog of (vector) autoregression for time series data (e.g., \citealp{kowal2017functional}). 

There are several important special cases of the DFOSR model \eqref{fts}-\eqref{evol}. If $\bm x_t = \bm 0$ for all $t$, i.e., there are no predictors, model \eqref{fts}-\eqref{reg} is a reduced-rank functional factor model with autocorrelated factors, which is useful for modeling and forecasting functional time series data \citep{FDFM,aue2015prediction,kowal2017bayesian}. If $\alpha_{j,k,t}=  \alpha_{j,k}$ for all $t$ and $\phi_k =0$ for all $k$, model \eqref{fts}-\eqref{reg} is a (Bayesian) FOSR model \citep{morris2006wavelet,zhu2011robust,montagna2012bayesian}. If $\alpha_{j,k,t}=  \alpha_{j,k}$ for all $t$, model \eqref{fts}-\eqref{reg} is a Bayesian FOSR model with autoregressive errors (FOSR-AR). Note that our setting is similar to, but distinct from, longitudinal functional data analysis (e.g., \citealp{greven2011longitudinal,park2015longitudinal}). Longitudinal functional data are time-ordered functional data, but typically include replicates of each functional time series (e.g., across subjects) and shorter time series. As a result, methodology for longitudinal functional data may incorporate autocorrelation, but relies less on the dynamic adaptability of \eqref{reg} and \eqref{evol}.


\section{Modeling the Loading Curves}\label{loadings}
Within-curve dependence of the functional data $Y_t$ is modeled by $\{f_k\}$ in \eqref{fts}. Existing methods for FOSR commonly rely on similar expansions in a (known or unknown) basis $\{f_k\}$. Notably, the dimensionality of the basis $K$ governs the dimensionality of the regression in \eqref{reg}. Methods that use full basis expansions, such as splines \citep{laurini2014dynamic} or wavelets \citep{morris2006wavelet,zhu2011robust}, are neither parsimonious nor computationally scalable in the presence of other dependence, such as autocorrelated functional data or time-varying regression functions. An alternative approach is to pre-compute a lower-dimensional basis, such as in functional principal components analysis (FPCA);  see \cite{goldsmith2016assessing}. However, methods that pre-compute a functional basis fail to account for the uncertainty in the unknown basis. This uncertainty is nontrivial: \cite{goldsmith2013corrected} demonstrate that FPC-based methods may substantially underestimate total variability, even for densely-observed functional data. 

Several existing Bayesian reduced-rank functional data models do account for the uncertainty in the dimension reduction, but in general lack sufficient computational scalability (see Table \ref{table:comp}) and model flexibility. \cite{suarez2017bayesian} propose a Bayesian FPCA, but do not incorporate predictors or dependence structures, and rely on a computationally expensive reversible-jump MCMC. \cite{montagna2012bayesian} incorporate predictors, but the model is non-dynamic and does not include shrinkage priors to reduce the impact of unimportant variables. \cite{kowal2017bayesian} propose a functional dynamic linear model, but do not use shrinkage priors for the (time-varying) regression coefficients, which results in less accurate estimates with larger variability (see Section \ref{simulations}). 
In addition, \cite{kowal2017bayesian} only consider functional data with univariate observation points ($D=1$), which limits applicability.

We propose a model for the loading curves $\{f_k\}$ that simultaneously (i) treats $\{f_k\}$ as unknown, which produces a data-adaptive basis and minimizes the number of necessary basis functions $K$; (ii) accounts for the inherent uncertainty in $\{f_k\}$; (iii) is scalable in the number of observation points, $M$; and (iv) is well-defined for $\mathcal{T}\subset \mathbb{R}^D$ with $D \in \mathbb{Z}^+$. In particular, we model each $f_k$ using low-rank thin plate splines (LR-TPS), which are smooth, flexible, and known to be efficient in MCMC samplers \citep{crainiceanu2005bayesian}. We present a general approach for arbitrary basis expansions, but provide details for our preferred LR-TPS implementation in the Appendix. 



\subsection{Full Conditional Distributions: General Basis Functions}\label{genBasis}
A common approach in nonparametric regression and functional data analysis is to represent each unknown function---here, each $f_k$---as a linear combination of known basis functions, and then model the corresponding unknown basis coefficients. Let $f_k(\bm \tau) = \bm b'(\bm \tau) \bm \psi_k$, where $\bm b'(\bm\tau) = (b_1(\bm\tau), \ldots, b_{L_M}(\bm\tau))$ is an $L_M$-dimensional vector of \emph{known} basis functions and $\bm\psi_k$ is an $L_M$-dimensional vector of \emph{unknown} basis coefficients. Popular choices for $\bm b(\cdot)$ include splines, Fourier basis functions, wavelets, and radial basis functions \citep{silverman2005functional,morris2015functional}. The choice of basis functions may be application-specific, and the number of basis functions $L_M$ may depend on the selected basis and the number of observation points, $M$; we provide default specifications for LR-TPS in the Appendix. Typically, basis expansions are combined with a suitable penalty function, such as $\mathcal{P}(f_k) = \int \left[\ddot{f_k}( \tau)\right]^2 d\tau$ for $\ddot f_k$ the second derivative of $f_k$ (assuming $D=1$), which encourages smoothness and guards against overfitting. For Bayesian implementations, such penalties correspond to prior distributions on the basis coefficients $\bm \psi_k$, or equivalently, the implied function $f_k$. For example, the roughness penalty above may be written $\mathcal{P}(f_k) = \bm \psi_k' \bm \Omega_{b} \bm \psi_k$ for known $L_M \times L_M$ penalty matrix $\bm \Omega_{b}$ with $(\ell, \ell')$ entry $[\bm \Omega_{b}]_{\ell, \ell'} = \int \ddot{b}_\ell(\tau)\ddot{b}_{\ell'}(\tau)d\tau$, which is commonly expressed as $\bm \psi_k \sim N(\bm 0, \lambda_{f_k}^{-1} \bm \Omega_{b}^{-1})$ for smoothing parameter $\lambda_{f_k} > 0$. For generality, we assume the prior $\bm \psi_k \sim N(\bm 0, \bm \Sigma_{\psi_k})$ for $k=1,\ldots,K$, which implies a Gaussian process prior on $f_k$ with mean function zero and covariance function $\mbox{Cov}(f_k(\bm\tau), f_k(\bm u)) = \bm b'(\bm\tau) \bm \Sigma_{\psi_k} \bm b(\bm u)$. 


Given functional data observations $\bm Y_t = (Y_t(\bm \tau_1),\ldots, Y_t(\bm \tau_M))'$ at observation points $\{\bm \tau_j\}_{j=1}^M$, the likelihood in \eqref{fts} becomes 
\begin{equation}\label{ftsVec}
\bm Y_t = \sum_{k=1}^K \bm f_k \beta_{k,t} + \bm \epsilon_t,
\quad \bm \epsilon_t \stackrel{indep}{\sim} N(\bm 0, \sigma_{\epsilon_t}^2 \bm I_M)
\end{equation}
where  $\bm f_k =  (f_k(\bm \tau_1),\ldots,  f_k(\bm \tau_M))' = \bm B \bm \psi_k$ are the loading curves evaluated at the observation points, with $\bm B = (\bm b(\tau_1), \ldots, \bm b(\tau_M))'$ the $M\times L_M$ basis matrix. We construct a Bayesian backfitting sampling algorithm that iteratively draws from the full conditional distribution of each $f_k$ conditional on $\{f_\ell\}_{\ell \ne k}$. The full conditional distribution of the corresponding basis coefficients is  $\left[\bm \psi_k | \cdots\right] \sim N\left(\bm Q_{\psi_k}^{-1} \bm \ell_{\psi_k}, \bm Q_{\psi_k}^{-1}\right)$, where 
$\bm Q_{\psi_k} =  (\bm B'\bm B)\sum_{t=1}^T \left(\beta_{k,t}^2/\sigma_{\epsilon_t}^2\right) +  \bm \Sigma_{\psi_k}^{-1}$ and $\bm \ell_{\psi_k} =  \bm B' \sum_{t=1}^T \big[\left(\beta_{k,t}/\sigma_{\epsilon_t}^2\right) \big( \bm Y_t - \sum_{\ell \ne k} \bm f_\ell \beta_{\ell,t} \big)\big]
$. 
Sampling $\bm \psi_k$ has computational complexity at most $\mathcal{O}(L_M^3)$. By comparison, a full rank Gaussian process has computational complexity $\mathcal{O}(M^3)$, and further requires computation of the inverse $\bm \Sigma_{\psi_k}^{-1}$. For  LR-TPS,  $\bm \Sigma_{\psi_k}^{-1} = \lambda_{f_k}\bm \Omega_{b}$ for known matrix $\bm \Omega_{b}$, which eliminates a matrix inversion. If each $f_k$ is smooth, we may let $L_M \ll M$ to improve scalability without sacrificing model fit.


\subsection{Simplifying the Likelihood via Identifiability Constraints}\label{orthog}
We enforce identifiability constraints on the loading curves, $\{f_k\}$, which primarily serves two purposes. First, identifiability allows us to interpret $\{f_k\}$  and the $k$-specific model parameters in \eqref{reg} and \eqref{evol}. Second, our particular choice of constraints provides computational improvements for sampling the parameters in \eqref{reg} and \eqref{evol}. We constrain $\bm F'\bm F = \utwi I_K$, where $\bm F = (\bm f_1,\ldots, \bm f_K) $ is the $M \times K$ matrix of loading curves evaluated at the observation points $\bm\tau_1,\ldots,\bm\tau_M$ and $\utwi I_K$ is the $K\times K$ identity matrix. This constraint, combined with a suitable ordering constraint on $k=1,\ldots,K$ (see Section \ref{shrinkage}), is sufficient for identifiability (up to sign changes, which in our experience are not problematic in the MCMC sampler).

The utility of our orthonormality constraint is illustrated with the following result: 

\begin{lemma}\label{ftsLike2}
Under the identifiability constraint $\bm F' \bm F = \bm I_K$, the joint likelihood in \eqref{ftsVec} for $\{\beta_{k,t}\}$ is equivalent to the \emph{working likelihood} implied by
\begin{equation}\label{ftsVec2}
 \tilde Y_{k,t} =  \beta_{k,t} + \tilde \epsilon_{k,t},
\quad\tilde \epsilon_{k,t} \stackrel{indep}{\sim} N(0, \sigma_{\epsilon_t}^2)
\end{equation}
up to a constant that does not depend on $\beta_{k,t}$, where $ \tilde Y_{k,t} = \bm f_k' \bm Y_t $ and $\tilde \epsilon_{k,t} = \bm f_k' \bm \epsilon_t$. 
\end{lemma}
For sampling the factors $\beta_{k,t}$ (and associated parameters), we only need the likelihood \eqref{ftsVec2}, which only depends on $M$ via the projection $ \tilde Y_{k,t} = \bm f_k' \bm Y_t$. The projection step is a one-time cost (per MCMC iteration). As a result, the model complexity for the dynamic components in \eqref{reg} and \eqref{evol} is \emph{not} severely limited by the dimension of the functional data, $M$, nor the correlations among the components of $\bm Y_t$, which are often large for functional data. These computational simplifications afford us the ability to incorporate the complex dynamics in \eqref{reg}-\eqref{evol} without sacrificing computational feasibility (see Section \ref{MCMC} for an example).  

 As an empirical illustration, Table \ref{table:comp} gives computation times for simulated data from Section \ref{simulations} for the proposed DFOSR model compared to \cite{kowal2017bayesian} (defined as DFOSR-NIG in Section \ref{simulations}). Notably, \cite{kowal2017bayesian}  use a similar model for $\{f_k\}$,  but do \emph{not} use the identifiability constraint $\bm F' \bm F = \bm I_K$ to produce the simplifications in Lemma \ref{ftsLike2}. The improvements are substantial, particularly for the larger sample size.

\begin{table}[h!]
\centering
\begin{tabular}{ c | c | c }
MCMC Algorithm  & $T=50, M = 20$ &  $T=200, M = 100$\\ 
  \hline
 Proposed DFOSR & 48 seconds  & 3 minutes \\  
 \cite{kowal2017functional} & 15 minutes & 74 minutes
\end{tabular}
\caption{Computing times per 1000 MCMC iterations (implemented in \texttt{R} on a MacBook Pro, 2.7 GHz Intel Core i5). In all cases, $p=15$ and $K=6$. }
\label{table:comp}
\end{table}


For each  $\bm f_k$, the orthonormality constraint may be decomposed into two sets of constraints: the linear constraints $\bm f_\ell' \bm f_k = 0$ for $\ell \ne k$ and the unit-norm constraint, $|| \bm f_k ||^2 = 1$. Since the sampler in Section \ref{genBasis} conditions on $\{\bm f_\ell\}_{\ell \ne k}$, the linearity constraint is fixed for each $\bm f_k = \bm B \bm \psi_k$. Therefore, given the full conditional distribution $\left[\bm \psi_k | \cdots\right] \sim N\left(\bm Q_{\psi_k}^{-1} \bm \ell_{\psi_k}, \bm Q_{\psi_k}^{-1}\right)$, we enforce the linear orthogonality constraint by \emph{conditioning} on $\bm C_k \bm \psi_k = \bm 0$, where $\bm C_k = (\bm f_1, \ldots, \bm f_{k-1}, \bm f_{k+1}, \ldots, \bm f_K)' \bm B = (\bm \psi_1, \ldots, \bm \psi_{k-1}, \bm \psi_{k+1}, \ldots, \bm \psi_K)' \bm B' \bm B$. Conditioning on the constraint is particularly interpretable in a Bayesian setting, and produces desirable optimality properties for constrained penalized regression (see Theorem 1 of \citealp{kowal2017bayesian}). Since the full conditional distribution for $\bm \psi_k$ is Gaussian, conditioning on $\bm C_k \bm \psi_k = \bm 0$ produces a Gaussian distribution with easily computable mean and covariance. Sampling from the constrained distribution is straightforward and efficient: given a draw from the unconstrained posterior, say $\bm\psi_k^0 \sim N\left(\bm Q_{\psi_k}^{-1} \bm \ell_{\psi_k}, \bm Q_{\psi_k}^{-1}\right)$, we retain the vector $\bm \psi_k^* = \bm \psi_k^0 - \bm Q_{\psi_k}^{-1} \bm C_k' \left(\bm C_k\bm Q_{\psi_k}^{-1}\bm C_k'\right)^{-1}\bm C_k \bm \psi_k^0$.
Given the orthogonally-constrained sample $\bm f_k^* = \bm B \bm \psi_k^*$, we rescale to enforce the unit-norm constraint: $\bm f_k = \bm f_k^*/||\bm f_k^*||$, and similarly rescale $\bm \psi_k^*$. This rescaling does not change the shape of the loading curve $\bm f_k$, and can be counterbalanced by an equivalent rescaling of the corresponding factor, i.e., $\beta_{k,t} \leftarrow \beta_{k,t} ||\bm f_k^*||$, to preserve exactly the likelihood \eqref{ftsVec}. By applying this procedure iteratively for $k=1,\ldots,K$, the  constraint $\bm F'\bm F = \bm I_K$ is satisfied for 
 \emph{every} MCMC iteration.

\section{Shrinkage Priors for the Model} \label{shrinkage}
While the DFOSR \eqref{fts}-\eqref{evol} is highly flexible, it is also overparametrized: it is unlikely that the regression coefficients $\alpha_{j,k,t}$ change substantially for \emph{all} times $t$, or that every predictor $x_{j,t}$ has a strong association with the functional response $Y_t$. Careful choices of priors for $\sigma_{\eta_{k,t}}^2$ and $\sigma_{\omega_{j,k,t}}^2$ offer shrinkage toward simpler models, which often improves estimation accuracy and reduces variability (see Section \ref{simulations}). We propose nested \emph{horseshoe priors} \citep{carvalho2010horseshoe} for shrinkage toward locally-static regression models with fewer predictors, and  \emph{multiplicative gamma process} priors \citep{bhattacharya2011sparse} for ordered shrinkage across factors $k=1,\ldots,K$, which reduces the sensitivity to the choice of $K$.  In (non-functional) time-varying parameter regression, shrinkage priors offer improvements in prediction and provide narrower posterior credible intervals (e.g., \citealp{kowal2017dynamic}). 

For the dynamic regression coefficient innovations $\omega_{j,k,t} \stackrel{indep}{\sim} N(0, \sigma_{\omega_{j,k,t}}^2)$, we encourage shrinkage at multiple levels with the following hierarchy of half-Cauchy distributions:
\begin{equation}\label{shrink}
\sigma_{\omega_{j,k,t}} \stackrel{ind}{\sim} C^+(0, \lambda_{j,k}), \quad \lambda_{j,k} \stackrel{ind}{\sim} C^+(0, \lambda_{j}), \quad \lambda_{j} \stackrel{ind}{\sim} C^+(0, \lambda_{0}), \quad \lambda_{0} \stackrel{ind}{\sim} C^+(0, 1/\sqrt{T-1})
\end{equation}
First, $\sigma_{\omega_{j,k,t}} \approx 0$ implies that $|\omega_{j,k,t}| \approx 0$, so $\alpha_{j,k,t} \approx \alpha_{j,k,t-1}$ is locally constant. Each $\alpha_{j,k,t}$ for predictor $j$ and factor $k$ may vary at any time $t$, but the prior encourages most changes to be approximately negligible, which implies fewer effective parameters in the model. The shrinkage parameters $\lambda_{j,k}$ and $\lambda_j$ are common for all times $t$, and provide factor- and predictor-specific shrinkage: for each predictor $j$, $\lambda_{j,k}$ allows some factors $k$ to be nonzero, while $\lambda_j$ operators as a group shrinkage parameter that may effectively remove predictor $j$ from the model.  Lastly, the global shrinkage parameter $\lambda_0$ controls the global level of sparsity, and is scaled by $1/\sqrt{T-1}$ following \cite{piironen2016hyperprior}. In the case of the non-dynamic FOSR and FOSR-AR models, we simply remove one level of the hierarchy: $\omega_{j,k,t} \stackrel{indep}{\sim} N(0, \lambda_{j,k}^2)$. The simulation analysis of Section \ref{simulations} clearly demonstrate the importance of these shrinkage priors, particularly for time-varying parameter regression.

The multiplicative gamma process (MGP) provides ordered shrinkage with respect to factor $k$, which suggests that factors with larger $k$ explain less variability in the data, and effectively reduces sensitivity to the choice of the total number of factors, $K$. We assume MGP priors for the intercept terms $\mu_k \stackrel{indep}{\sim} N(0, \sigma_{\mu_k}^2)$, which are given by the prior on the precisions, $\sigma_{\mu_k}^{-2} = \prod_{\ell \le k} \delta_{\mu_\ell}$, where $ \delta_{\mu_1} \sim \mbox{Gamma}(a_{\mu_1}, 1)$ and $ \delta_{\mu_\ell} \sim \mbox{Gamma}(a_{\mu_2}, 1)$ for $\ell > 1$. As discussed in \cite{bhattacharya2011sparse}  and  \cite{durante2017note}, selecting $a_{\mu_1} > 0$ and $a_{\mu_2} \ge 2$ produces stochastic ordering among the implied variances $\sigma_{\mu_k}^2$, which also satisfies the ordering requirement for model identifiability. Similarly, for the innovations $\eta_{k,t} \sim N(0, \sigma_{\eta_{k,t}}^2)$ we follow  \cite{bhattacharya2011sparse} and \cite{montagna2012bayesian} and let  $\sigma_{\eta_{k,t}}^2 = \sigma_{\eta_k}^2/\xi_{\eta_{kt}}$  with $\sigma_{\eta_k}^{-2} = \prod_{\ell \le k} \delta_{\eta_\ell}$, $ \delta_{\eta_1} \sim \mbox{Gamma}(a_{\eta_1}, 1)$, $ \delta_{\eta_\ell} \sim \mbox{Gamma}(a_{\eta_2}, 1)$ for $\ell > 1$, and $\xi_{\eta_{kt}} \stackrel{iid}{\sim} \mbox{Gamma}(\nu_\eta/2, \nu_\eta/2)$. We allow the data to determine the rate of ordered shrinkage separately for $\{\mu_k\}$ and $\{\eta_{k,t}\}$ using the hyperpriors $a_{\mu_1}, a_{\mu_2}, a_{\eta_1},a_{\eta_2} \stackrel{iid}{\sim}\mbox{Gamma}(2,1)$. Finally, the hyperprior $\nu_\eta \sim \mbox{Uniform}(2, 128)$ for the degrees of freedom parameter incorporates the possibility of heavy tails in the marginal distribution for $\eta_{k,t}$.


\section{Simulations}\label{simulations}

\subsection{Simulation Design}
We conducted an extensive simulation study in order to characterize the performance of the proposed methods relative to state-of-the-art alternatives for functional regression and assess the relative importance of our modeling choices, including the model for the loading curves in \eqref{fts}, the time-varying parameter regression in \eqref{reg}-\eqref{evol}, and the shrinkage priors in Section \ref{shrinkage}. We consider simulation designs with \emph{dynamic} and \emph{non-dynamic} regression coefficients and different sample sizes: a small sample with $T = 50$ time points and $M = 20$ observation points, and a large sample with $T = 200$  and $M = 100$. 

We incorporate two sources of sparsity in the regression: (i) some predictors are not associated with the functional response $Y_t(\bm\tau)$ and (ii) some predictors are associated with $Y_t(\bm \tau)$ exclusively via a small number of factors. We fix $p_0 = 10$ regression coefficients to be exactly zero (for all times $t$), and let $p_1 = 5$ be nonzero, resulting in $p = p_0 + p_1 =  15$ regression coefficients (plus an intercept). For each nonzero predictor $j=1,\ldots, p_1 = 5$,  we uniformly sample $p_j^*$ factors to be nonzero, where $p_j^* \stackrel{iid}{\sim} \mbox{Poisson}(1)$ truncated to $[1, K^*]$. For dynamic regression coefficients, we simulate the nonzero factors $k$ for predictor $j$ from a Gaussian random walk with randomly selected jumps:  
$\alpha_{j,k,t}^* = Z_{k,0} + \sum_{s \le t} Z_{k,s} I_{k,s}$
where $Z_{k,t} \stackrel{indep}{\sim} N(0, 1/k^2)$ and $I_{k,t}  \stackrel{iid}{\sim} \mbox{Bernoulli}(0.01)$, which results in  time-varying yet locally constant regression coefficients $\alpha_{j,k,t}^*$. For non-dynamic regression coefficients, we simulate $\alpha_{j,k}^*\stackrel{indep}{\sim} N(0, 1/k^2)$. For all cases, the predictors are simulated from $x_{j,t} \stackrel{iid}{\sim}N(0,1)$, and the intercepts are fixed at $\mu_k^* = 1/k$. Finally, the autoregressive errors are $\gamma_{k,t}^* = 0.8 \gamma_{k,t-1}^* +\eta_{k,t}^*$ and $\eta_{k,t}^*  \stackrel{indep}{\sim}N(0, [1 - 0.8^2]/k^2)$, which are highly correlated yet stationary with marginal standard deviation $1/k$.

For $M$ equally-spaced points $\bm \tau \in [0,1]$, the true loading curves are $f_1^*(\bm\tau) = 1/\sqrt{M}$ and for $k = 2,\ldots, K^* = 4$, $f_k^*$ is an orthogonal polynomial of degree $k$. Given true factors $\beta_{k,t}^* = \mu_k^* + \sum_{j=1}^p  x_{j,t} \alpha_{j,k,t}^* + \gamma_{k,t}^*$  and loading curves $f_k^*(\bm \tau)$, the true curves are $Y_t^*(\bm \tau) = \sum_{k=1}^{K^*} f_k^*(\bm \tau) \beta_{k,t}^*$ and the functional data are simulated from $Y_t(\bm \tau) = Y_t^*(\bm \tau) + \sigma^* \epsilon_t^*(\bm\tau)$, where $\epsilon_t^*(\bm\tau)\stackrel{iid}{\sim}N(0, 1)$. After selecting a \emph{root-signal-to-noise ratio} (RSNR),  the observation error standard deviation is $\sigma^* =  \sqrt{\frac{\sum_{t=1}^T \sum_{j=1}^M (Y_t^*(\bm \tau_j) - \bar Y^*)^2}{TM - 1}}\Big/\mbox{RSNR}$ where $\bar Y^*$ is the sample mean of $\{Y_t^*(\bm \tau_j)\}_{j,t}$. We select RNSR = 5, which produces moderately noisy functional data.

\subsection{Methods For Comparison}
We consider two variations of the proposed methodology: the  DFOSR model \eqref{fts}-\eqref{evol}  (DFOSR-HS) and the non-dynamic analog with 
$\alpha_{j,k,t}=  \alpha_{j,k}$ (FOSR-AR), both with $K = 6 > K^*  = 4$ to include more factors than necessary. We consider an alternative DFOSR model with normal-inverse-gamma innovations (DFOSR-NIG), i.e., we replace the horseshoe priors in \eqref{shrink} with $\sigma_{\omega_{j,k}}^{-2} \stackrel{iid}{\sim} \mbox{Gamma}(0.001, 0.001)$. Originally proposed by  \cite{kowal2017bayesian}, this model does \emph{not} provide aggressive shrinkage with respect to time $t$, predictor $j$, or factor $k$, but otherwise retains the proposed DFOSR model characteristics. Next, to study the importance of estimating the loading curves $f_k$, we implement a variation of  DFOSR-NIG in which the loading curves $f_k$ are estimated {\it a priori} as functional principal components using  \cite{xiao2013fast}, where  $K$ is selected to explain 99\% of the variability in $\{Y_t^*(\bm \tau_j)\}_{j,t}$. For this method (Dyn-FPCA), we remove the ordered shrinkage by specifying $\mu_k \stackrel{iid}{\sim}N(0,100^2)$ and normal-inverse-gamma priors for $\eta_{k,t}$ in \eqref{reg} and $\omega_{j,k,t}$ in \eqref{evol}. 
Among existing FOSR methods, we include \cite{reiss2010fast}, which is a FOSR estimated using least squares (FOSR-LS), and \cite{barber2017function}, which is a FOSR with a group lasso penalty on each regression function (FOSR-Lasso), both implemented using the \texttt{refund} package in \texttt{R} \citep{refund}. 
These methods are non-Bayesian, and do not account for time-varying regression coefficients or autocorrelated errors (with respect to time).


\subsection{Simulation Results}
We compare methods using root mean squared errors of the dynamic regression coefficient functions, $\mbox{RMSE} = \sqrt{\frac{1}{pTM}\sum_{j=1}^p \sum_{t=1}^T \sum_{\ell = 1}^M (\tilde \alpha_{j,t}(\bm \tau_\ell) - \tilde \alpha_{j,t}^*(\bm \tau_\ell))^2}$, where $\tilde \alpha_{j,t}(\bm\tau_\ell)$ is the estimated regression coefficient for predictor $j$ at time $t$ and observation point $\bm \tau_\ell$ and $\tilde \alpha_{j,t}^*(\bm \tau_\ell) = \sum_{k=1}^{K^*} f_k^*(\bm\tau_\ell) \alpha_{j,k,t}^*$ is the true regression coefficient. For the Bayesian methods, we use the posterior expectation of $\tilde \alpha_{j,t}(\bm \tau_\ell)$ as our estimator. The RMSEs for the regression coefficients based on 50 simulations are in Figure \ref{fig:sims}. 

\begin{figure}[h]
\begin{center}
\includegraphics[width=.49\textwidth]{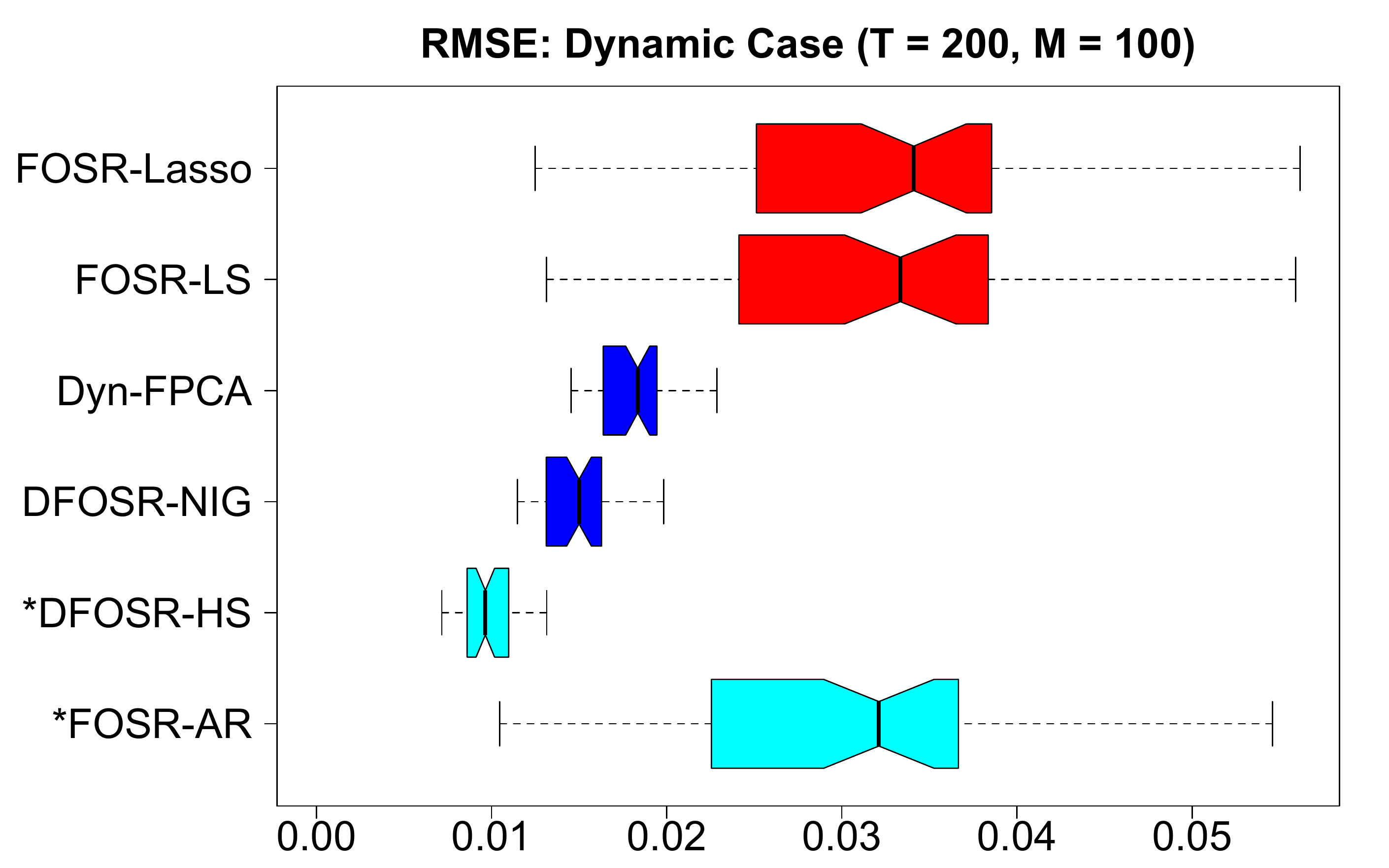}
\includegraphics[width=.49\textwidth]{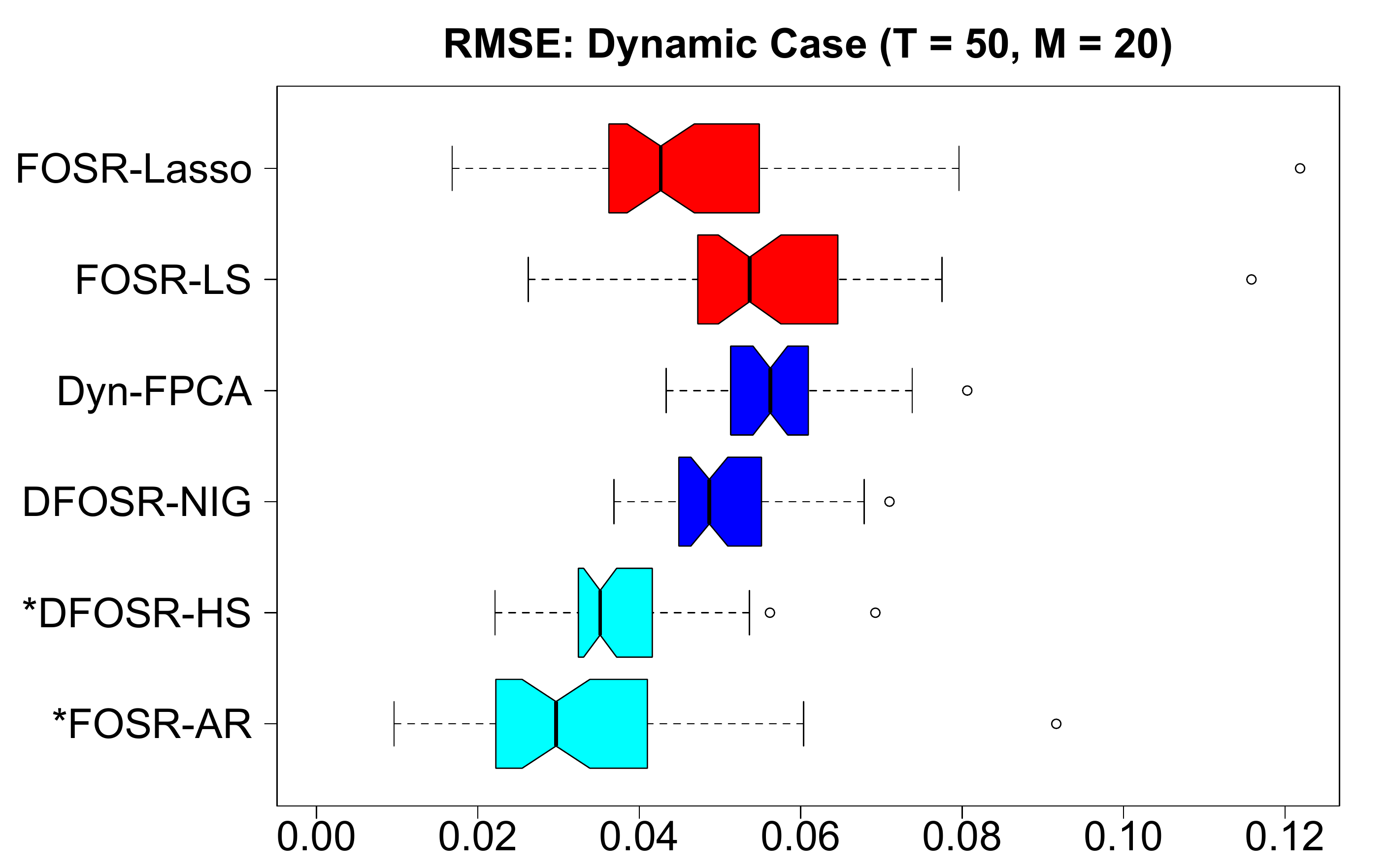}
\includegraphics[width=.49\textwidth]{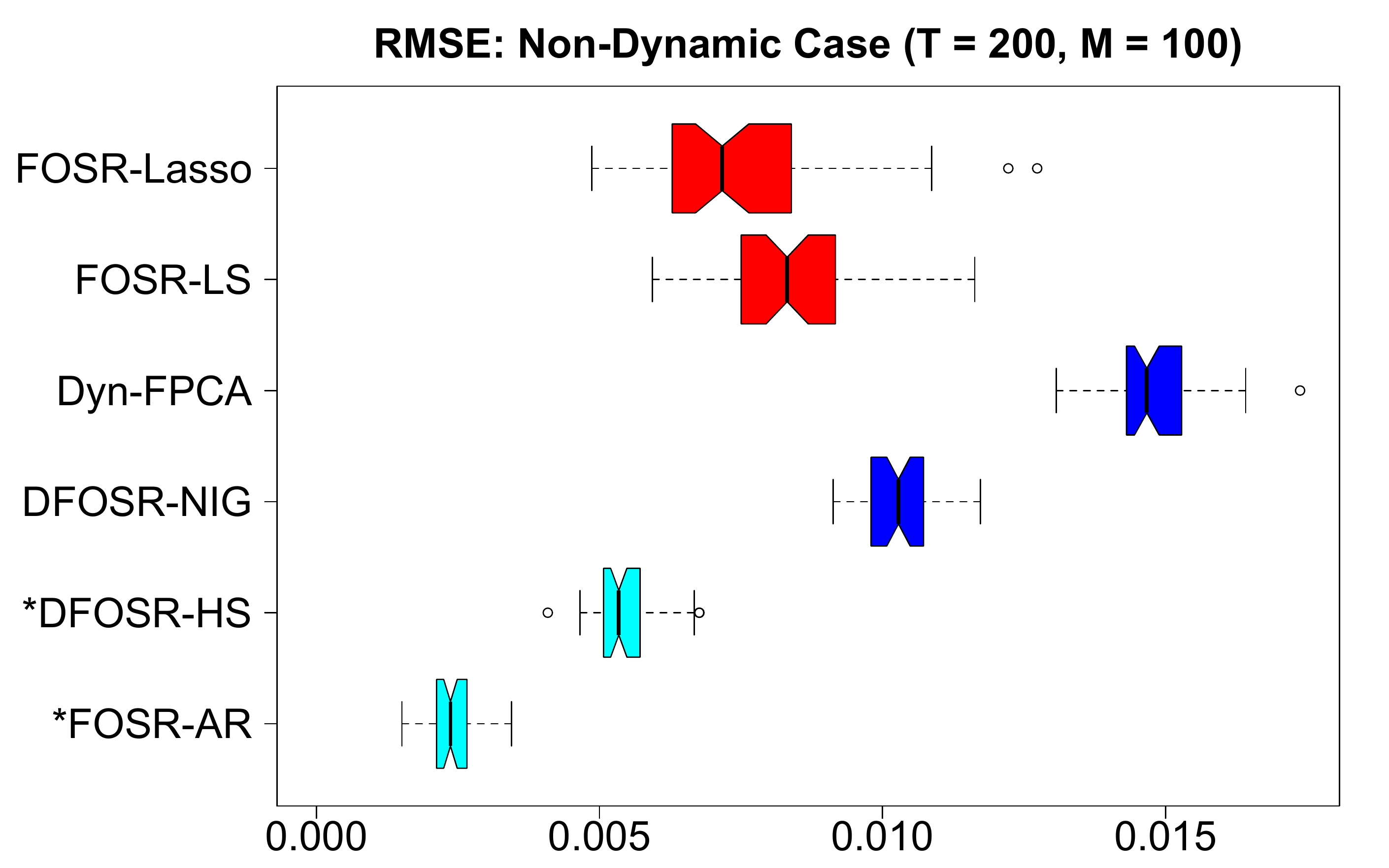}
\includegraphics[width=.49\textwidth]{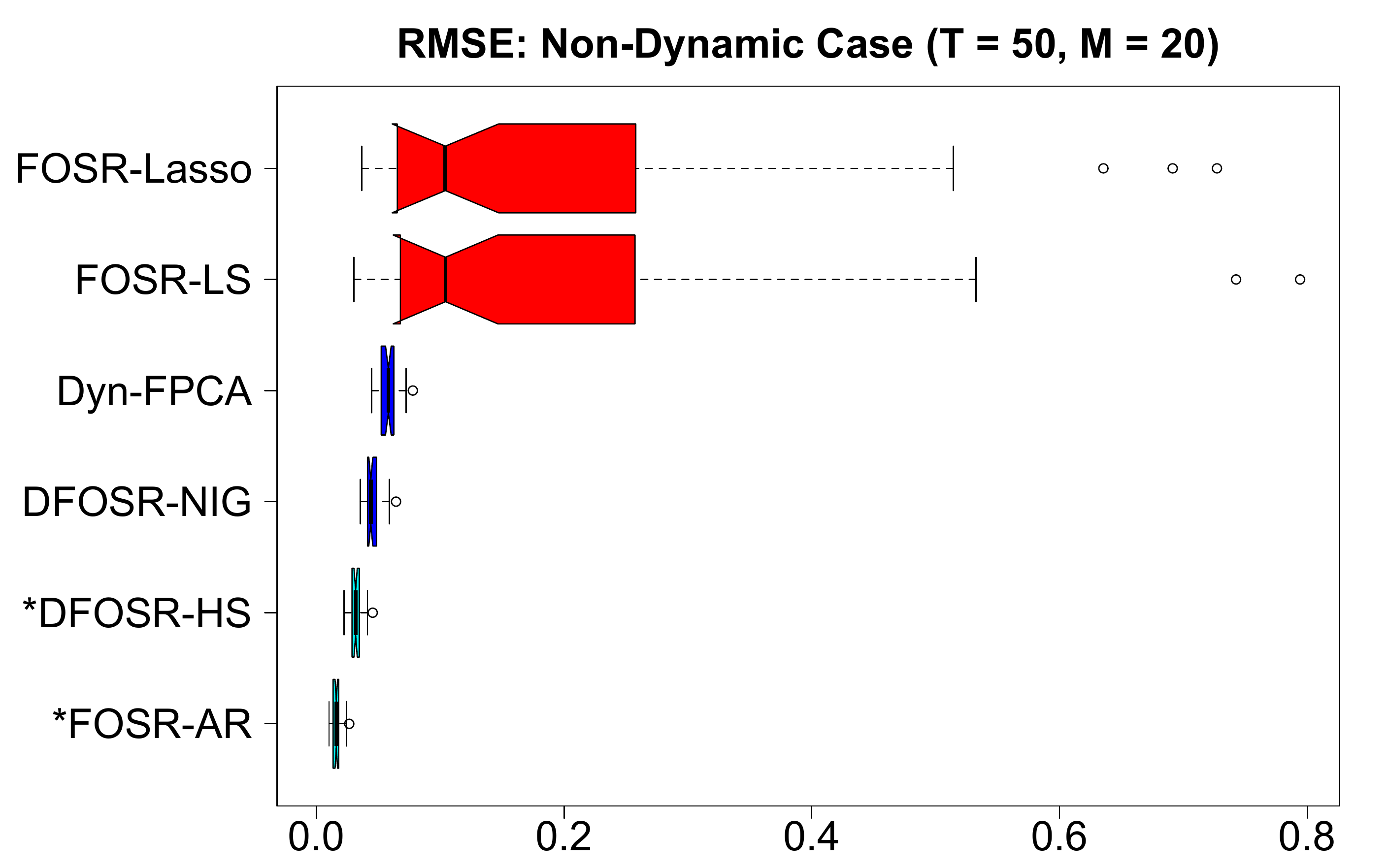}
\caption{Root mean squared errors for the regression coefficient functions $\tilde \alpha_{j,t}(\bm\tau)$ under different simulation designs: the dynamic case ({\bf top row}) and the non-dynamic case ({\bf bottom row}) for large {(\bf left column}) and small ({\bf  right column}) sample sizes. The proposed methods (DFOSR-HS and FOSR-AR) are marked with an asterisk and colored in light blue; simplifications of the proposed methods are in dark blue; and existing FOSR methods are in red. 
 \label{fig:sims}}
\end{center}
\end{figure}

In all cases, the proposed DFOSR-HS model performs better than existing methods, typically by a wide margin. Among time-varying parameter models, DFOSR-HS offers substantial improvements over DFOSR-NIG and Dyn-FPCA, which suggests that the shrinkage priors of Section \ref{shrinkage} are an important component of the DFOSR model. DFOSR-NIG is uniformly better than Dyn-FPCA, which demonstrates that our model for the loading curves $f_k$ in Section \ref{loadings} improves upon an FPCA-based approach. For the dynamic simulations, the comparative performance of these methods depends on the sample size: when $T=200$ and $M=100$, the time-varying parameter regression models (DFOSR-HS, DFOSR-NIG, and Dyn-FPCA) are clearly preferable, but when $T=50$ and $M=20$, only the proposed DFOSR-HS performs well among dynamic models, and the (non-dynamic) FOSR-AR performs best overall. For the non-dynamic simulations, FOSR-AR performs best followed by DFOSR-HS for both sample sizes. 

In addition, we compare \emph{mean credible interval widths} (MCIWs) for the time-varying parameter regression models (DFOSR-HS, DFOSR-NIG, and Dyn-FPCA) in Figure \ref{fig:sims-mciw}.  The MCIWs are defined as 
$ \mbox{MCIW} = \frac{1}{pTM}\sum_{j=1}^p \sum_{t=1}^T \sum_{\ell = 1}^M \left[\tilde \alpha_{j,t}^{(95)}(\bm \tau_\ell) - \tilde \alpha_{j,t}^{(5)}(\bm \tau_\ell)\right]
$ where $\tilde \alpha_{j,t}^{(95)}(\bm \tau_\ell)$ and  $\tilde \alpha_{j,t}^{(5)}(\bm \tau_\ell)$ are the 95\% and 5\% quantiles, respectively, of the posterior distribution for $\tilde \alpha_{j,t}(\bm \tau_\ell)$. In each case, the empirical coverage exceeds 96\%, which is more conservative than the 90\% nominal coverage. Notably, DFOSR-HS obtains substantially narrower credible intervals without sacrificing nominal coverage, which suggests greater power to detect functional associations.

\begin{figure}[h]
\begin{center}
\includegraphics[width=.49\textwidth]{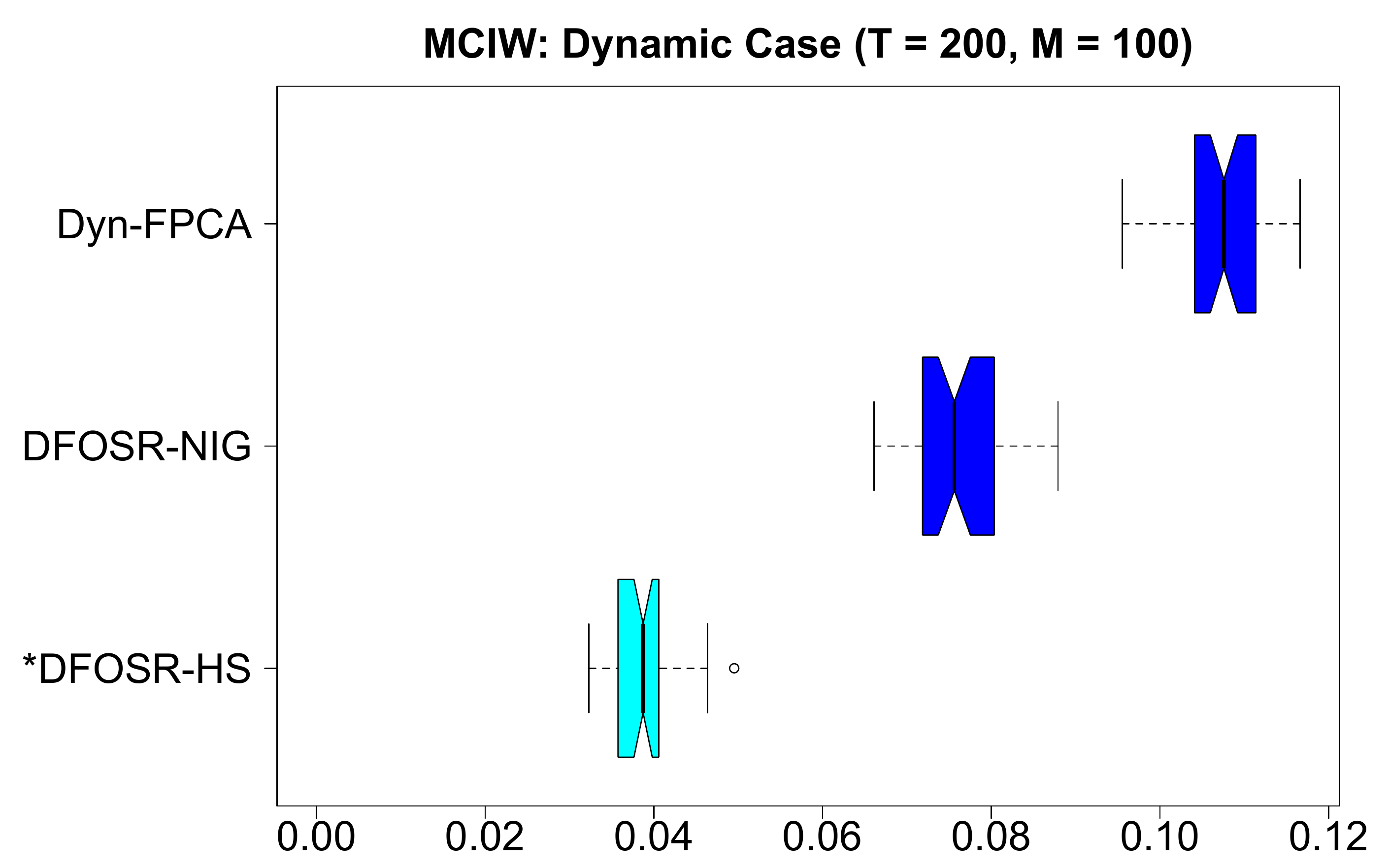}
\includegraphics[width=.49\textwidth]{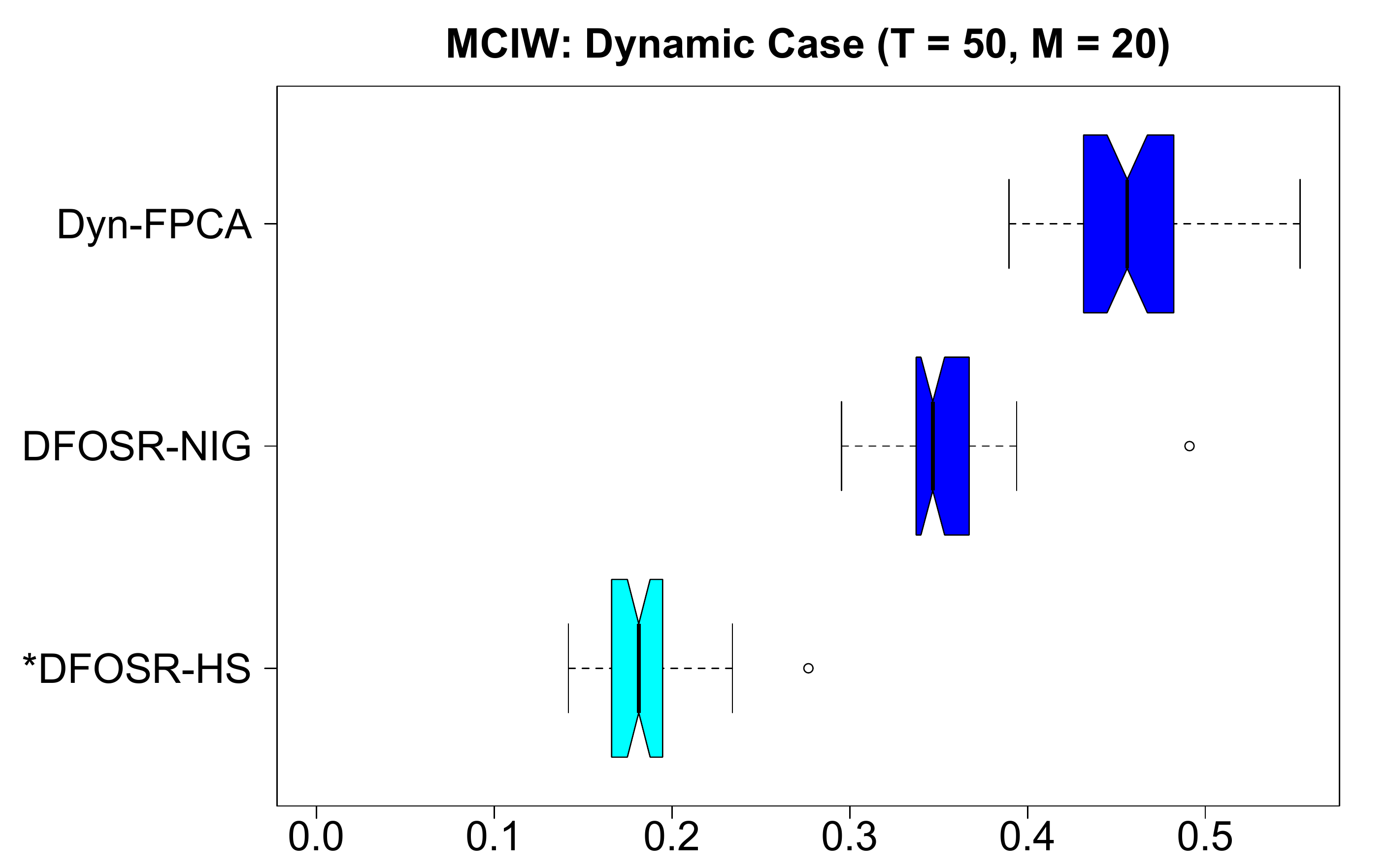}
\includegraphics[width=.49\textwidth]{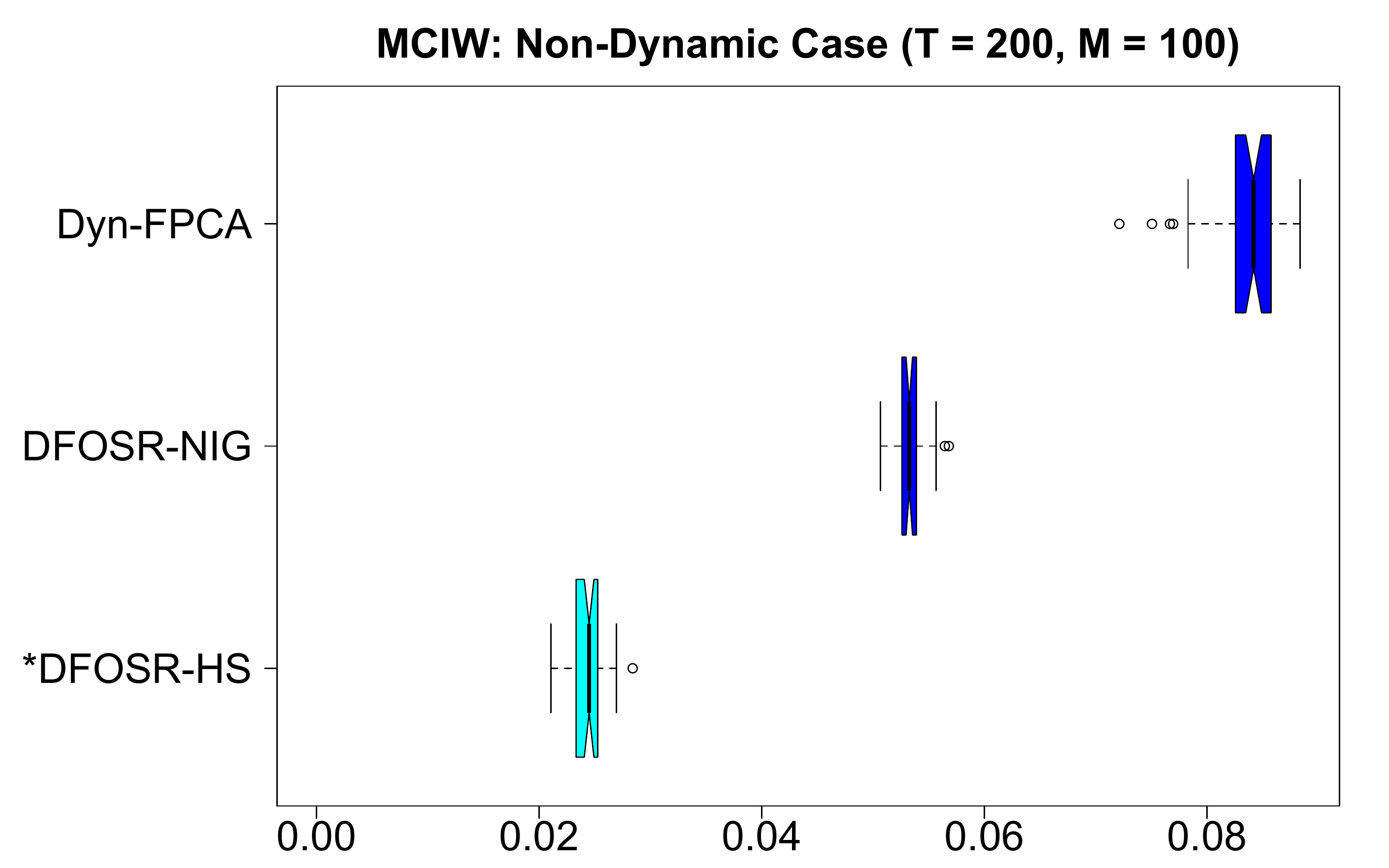}
\includegraphics[width=.49\textwidth]{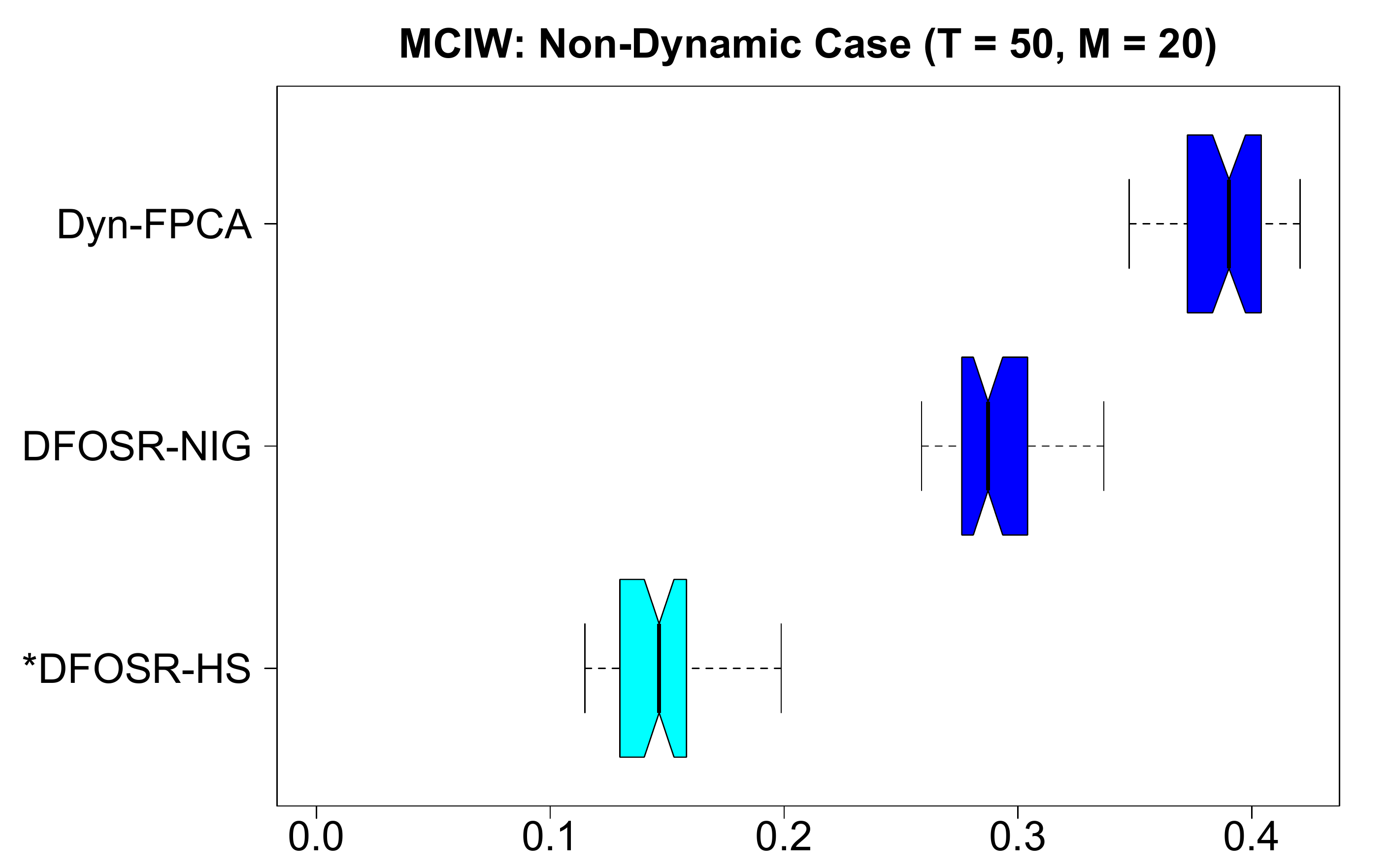}
\caption{Mean credible interval widths for the regression coefficient functions $\tilde \alpha_{j,t}(\bm\tau)$ under different simulation designs: the dynamic case ({\bf top row}) and the non-dynamic case ({\bf bottom row}) for large {(\bf left column}) and small ({\bf  right column}) sample sizes. The proposed method (DFOSR-HS) is marked with an asterisk. 
 \label{fig:sims-mciw}}
\end{center}
\end{figure}



\section{Macroeconomy and the Yield Curve}\label{yields}
The yield curves describes the time-varying term structure of interest rates: at each time $t$, the yield curve $Y_t(\bm\tau)$ characterizes how interest rates vary over the length of the borrowing period, or maturity, $\bm \tau$. Yield curves are an essential component in many economic and financial applications: they provide valuable information about economic and monetary conditions, inflation expectations, and business cycles, and are used to price fixed-income securities  and construct forward curves \citep{bolder2004empirical}. Due to these fundamental economic connections, we are interested in the associations between the yield curve and key macroeconomic variables, namely, real activity, monetary policy, and inflation. Importantly, the DFOSR modeling framework allows us to associate these variables with particular maturities $\bm\tau$ along the yield curve, and to study how the associations may change over time. 

Dynamic yield curve models commonly adopt the Nelson-Siegel parameterization \citep{nelson1987parsimonious}, usually within a state space framework \citep{diebold2006forecasting,diebold2006macroeconomy,koopman2010analyzing}. These parametric approaches are less flexible and introduce bias in estimation and forecasting, and often require solving computationally intensive nonlinear optimization problems. Nonparametric methods include \cite{FDFM} and \cite{SDFM}, but these approaches do not provide the uncertainty quantification, time-varying parameter regression, and shrinkage capabilities of the DFOSR model.

We obtain zero-coupon U.S. yield curve data from \cite{gurkaynak2007us}, which are pre-smoothed using \cite{svensson1994estimating} for $M=30$ maturities $\bm \tau_j \in \mathcal{T}_{obs} \equiv \{1,\ldots, 30\}$ years. The macroeconomic predictors are manufacturing capacity utilization (CU; \url{https://fred.stlouisfed.org/series/TCU}) for real activity, the federal funds rate (FFR; \url{https://fred.stlouisfed.org/series/FEDFUNDS}) for monetary policy, and (annualized) price inflation (PCE; \url{https://fred.stlouisfed.org/series/PCEPI}) for inflation, which are centered and scaled. We compute monthly averages of the yield curve data for common frequency with the macroeconomic variables, and consider the time period from January 1986 to February 2018  ($T = 386$). 

Within the DFOSR model \eqref{fts}-\eqref{evol}, we include a stochastic volatility model for $\sigma_{\epsilon_t}^2$ to incorporate volatility clustering, which is an important component in many financial and economic applications (see the Appendix for details and a supporting figure). In addition, we impose stationarity via the AR coefficient priors $\left[(\phi_k + 1)/2\right] \stackrel{iid}{\sim} \mbox{Beta}(5, 2)$. We report results for $K=6$, but larger values of $K$ produce nearly identical results. We ran the MCMC algorithm of Section \ref{MCMC} for 16000 iterations, discarded the first 10000 simulations as a burn-in, and retained every 3rd sample. Traceplots indicate good mixing and suggest convergence.

In Figure \ref{fig:surf_yields}, we plot the posterior expectation of the dynamic regression functions $\tilde\alpha_{j,t}(\bm\tau) = \sum_{k=1}^K f_k(\bm\tau) \alpha_{j,k,t}$ for CU, FFR, and PCE for all times $t$ and maturities $\bm \tau$. During the late 1980s and 1990s, CU appears to impact the curvature of the yield curve, with a prominent hump for maturities around 10 years, but this effect dissipates during the 2000s.  FFR has the largest estimated effect, almost entirely for small maturities, which impacts the slope of the yield curve. Notably, the FFR effect is mostly time-invariant during this period (1986-2018). PCE has a moderate impact on the slope of the yield curve---in the opposite direction of FFR---but only until the 1990s. 

\begin{figure}[h]
\begin{center}
\includegraphics[width=.32\textwidth]{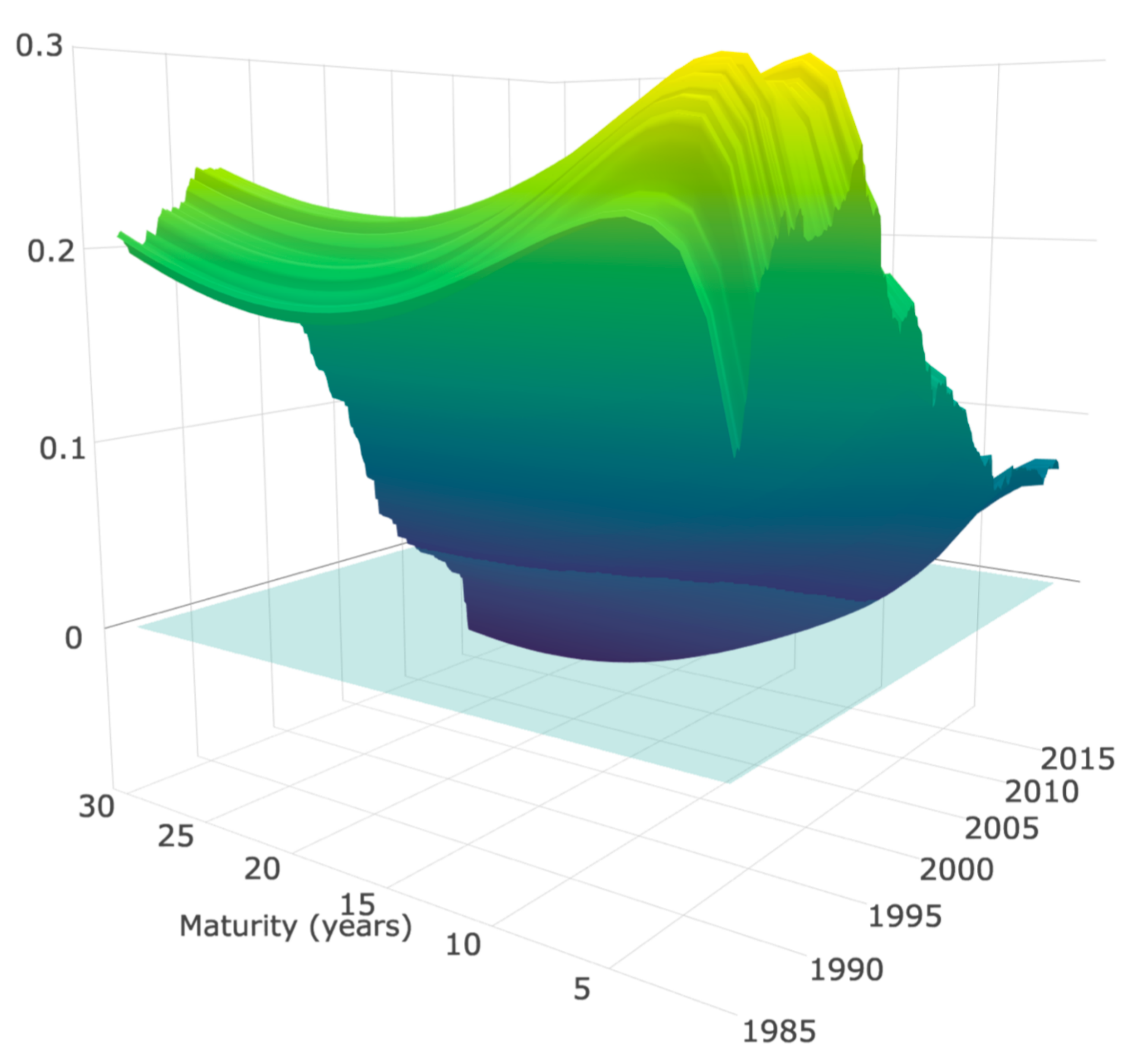}
\includegraphics[width=.32\textwidth]{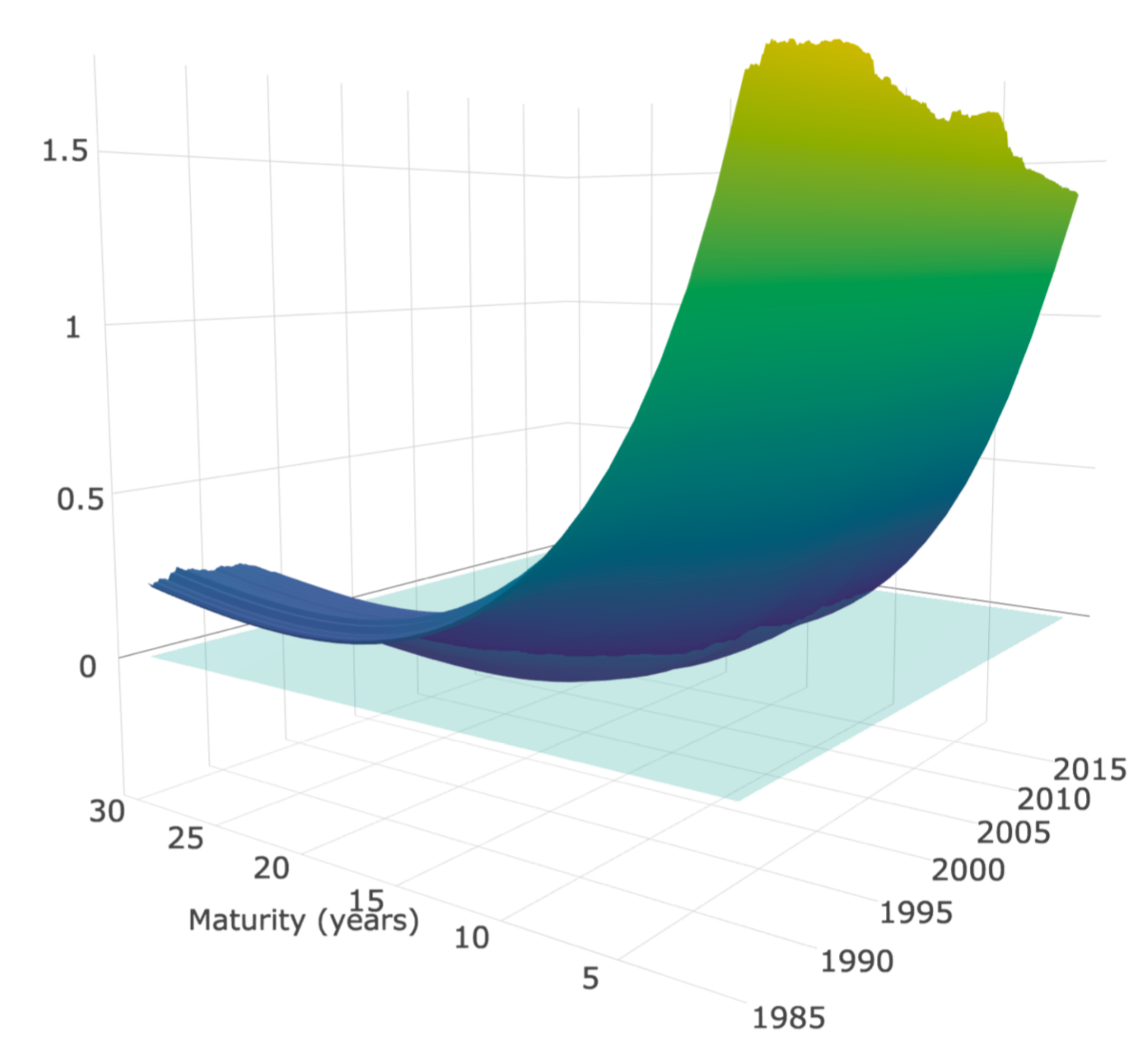}
\includegraphics[width=.32\textwidth]{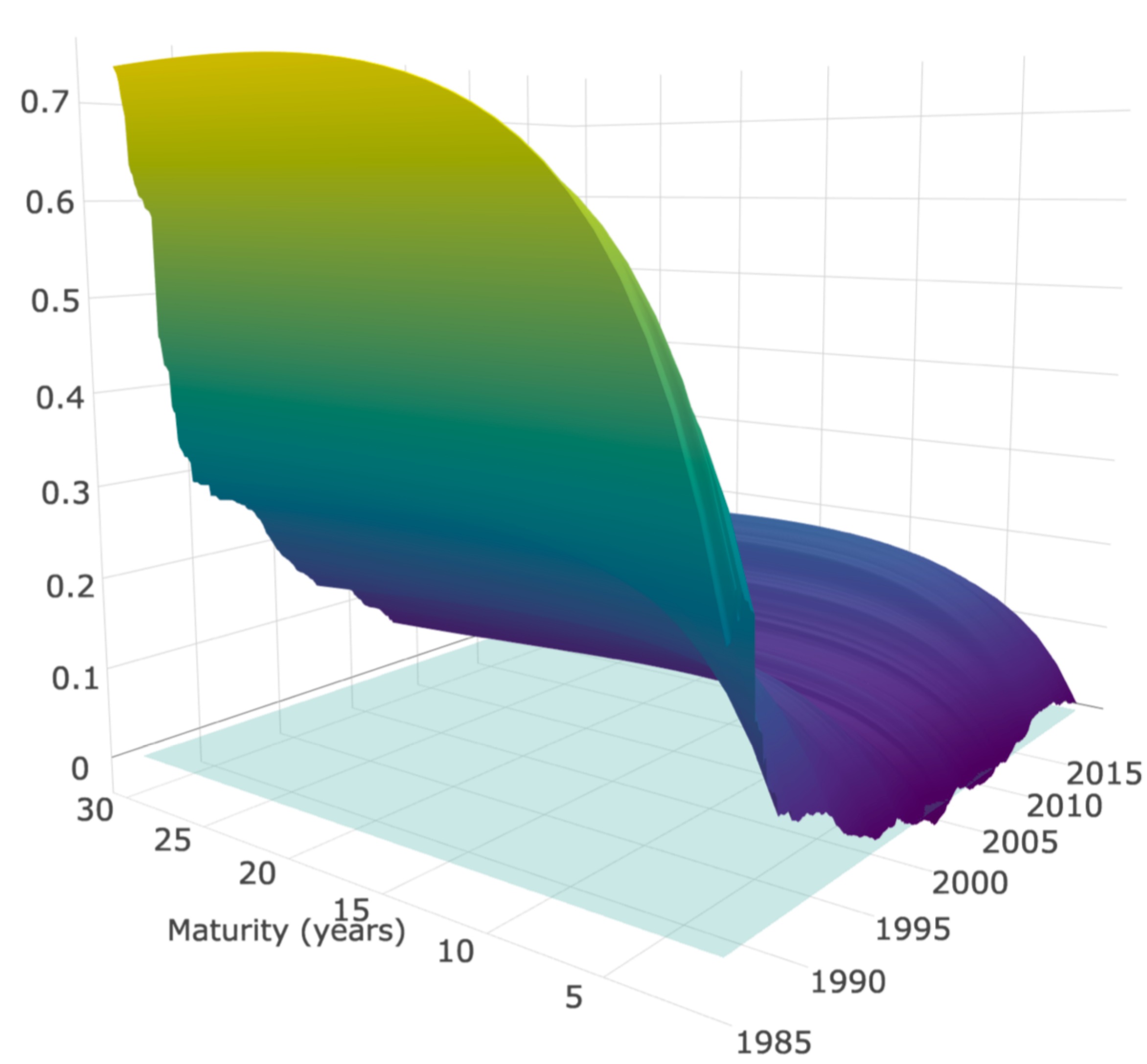}
\caption{Posterior expectation of the time-varying regression coefficient functions $\tilde\alpha_{j,t}(\bm\tau)$ for capacity utilization (CU, {\bf left}), federal funds rate (FFR, {\bf center}), and personal consumption expenditures (PCE, {\bf right}). The FFR has the largest estimated effect, particularly for smaller maturities. The impact of CU and PCE has declined substantially since the late 1980s. 
 \label{fig:surf_yields}}
\end{center}
\end{figure}

To further investigate these findings, Figure \ref{fig:curve_yields} presents the posterior expectations of $\tilde\alpha_{j,t}(\bm\tau)$ with 95\% pointwise credible intervals and simultaneous credible bands at select times $t$: March of 1986, 2002, and 2018. Naturally, the posterior expectations confirm the results in Figure \ref{fig:surf_yields}; however, the uncertainty quantification in Figure \ref{fig:curve_yields} offers additional insights. Notably, the width of the credible bands varies over time: the bands are widest in 1986 and most narrow in 2002, which reflects the dynamic adaptability of model \eqref{reg}-\eqref{evol} and the shrinkage priors of Section \ref{shrinkage}. The credible bands confirm the relative unimportance of CU as well as the clear association between FFR and yields for maturities of less than five years. Lastly, there is moderate evidence that PCE was associated with yields at longer maturities in 1986, but this effect vanished in more recent years. These results demonstrate the importance of incorporating both maturity-specific (functional) and time-varying (dynamic) effects in the model, which confirms the utility of the DFOSR model \eqref{fts}-\eqref{evol}.


\begin{figure}[h]
\begin{center}
\includegraphics[width=1\textwidth]{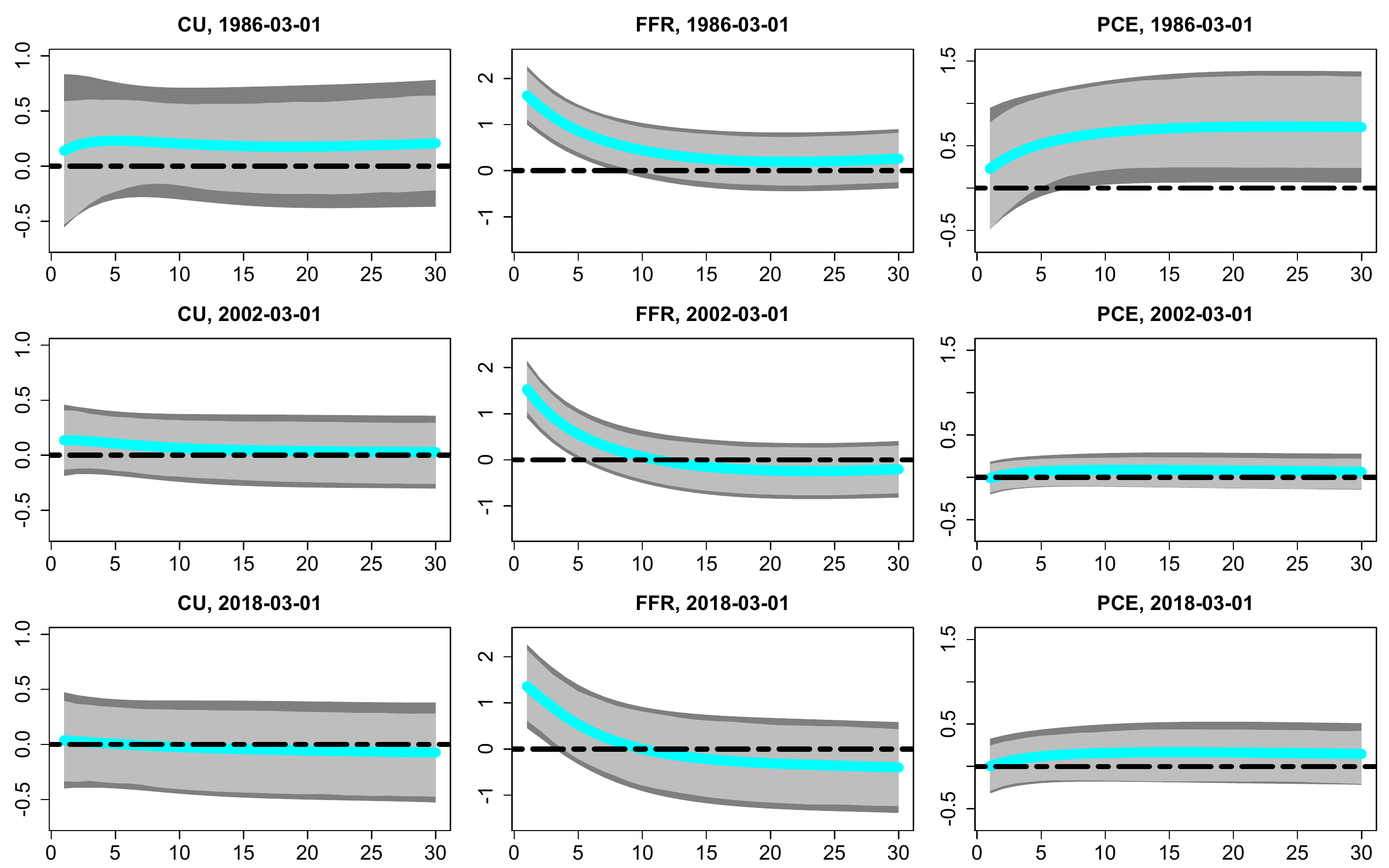}
\caption{Time-varying regression coefficient functions $\tilde\alpha_{j,t}(\bm\tau)$ as a function of maturity $\bm \tau$ (in years) for capacity utilization (CU, {\bf left}), federal funds rate (FFR, {\bf center}), and personal consumption expenditures (PCE, {\bf right}) in March of 1986 ({\bf top}), 2002 ({\bf middle}), and 2018 ({\bf bottom}). The posterior means (solid line) match the behavior in Figure \ref{fig:surf_yields}, but the posterior credible bands  (95\% pointwise intervals in light gray, 95\% simultaneous credible bands in dark gray) provide dynamic uncertainty quantification. 
 \label{fig:curve_yields}}
\end{center}
\end{figure}

\section{Age-Specific Fertility Rates in South and Southeast Asia}\label{fertility}
We analyze age-specific fertility rates (ASFRs) for developing nations in South and Southeast Asia. Fertility is an important determinant of the health and welfare of women, their families, and their communities, and is a key factor in global and national population growth. Fertility rates may vary greatly between developed and less developed nations, and may depend on socioeconomic and demographic factors such as age, education, employment, marital status, and access to family planning. While it is common for studies to use total fertility rates, which aggregate over all age groups, important patterns and trends in the fertility rate may only be discoverable using \emph{age-specific} fertility rates. The ASFR measures the annual number of births to women within a specific age group per 1000 women in that age group. Notably, equivalent total fertility rates may be attained using vastly different distributions of fertility among age groups (see \citealp{pantazis2018parsimonious}, Fig.\! 2). Naturally, the distribution of fertility among age groups is a fundamental determinant of future fertility rates and population sizes. Therefore, it is appropriate to model the ASFR as a functional time series: the fertility rate is a \emph{function} of age, and varies over \emph{time} (year).

A particular challenge in modeling ASFRs for developing nations is the sparsity of survey data. The Demographic and Health Surveys (DHS) of the United States Agency for International Development (USAID) aggregates available survey data, which may be accessed via STATcompiler \citep{casterline2010determinants}. We consider DHS survey data from 1994-2016 for 12 nations in South and Southeast Asia: Afghanistan, Bangladesh, Cambodia, India, Indonesia, Maldives, Myanmar, Nepal, Pakistan, Philippines, Timor-Leste, and Vietnam.  During this time period, four nations only have one available survey, and there are at most two surveys available each year; for years with two surveys, we use the average ASFRs. For each survey, the reported ASFR is the ASFR over the three years preceding the survey. 

The DHS survey data provides ASFRs for only a small number of age groups: 15-19, 20-24, 25-29, 30-34, 35-39, 40-44, 45-49. For modeling purposes, we use the midpoints of each age group, so the observation points are $\bm \tau_j \in \mathcal{T}_{obs} \equiv \{17, 22, 27, 32, 37, 42, 47\}$. Since we are interested in the age-specific fertility rates over the entire domain, $\mathcal{T} = [15, 49]$, we propose a model-based imputation approach to obtain estimates and inference for $M=31$ ages within the range of observed values: $\bm \tau = 17, \ldots,47$. In the Gibbs sampler, we draw $\left[Y_t(\bm \tau^*) | \{f_k\}, \{\beta_{k,t}\}, \sigma_{\epsilon}\right] \stackrel{indep}{\sim} N\big(\sum_k f_k(\bm \tau^*) \beta_{k,t}, \sigma_{\epsilon}^2\big)$ for each unobserved $\bm \tau^* \not\in \mathcal{T}_{obs}$, which provides (i) model-based interpolated fertility rate curves with posterior credible bands and (ii) inference for regression functions over a denser grid of points. 

In addition to the dynamic and functional aspects of ASFR data, we are interested in modeling the association between age-specific fertility and important socioeconomic and demographic predictor variables. In particular, we include the following predictor variables for each year $t$, provided by DHS and accessed via STATcompiler: (i) the percentage of currently married or in union women currently using any method of contraception, (ii) the median age of first marriage or union in years among women (age 25-49), (iii) the percentage of women with secondary or higher education, and (iv) the percentage of currently married or in union women employed in the 12 months preceding the survey. The  proposed DFOSR model provides a mechanism for understanding how each predictor impacts the \emph{shape} of the ASFR, with differential effects for different age groups.


Using the MCMC algorithm of Section \ref{MCMC}, we sample from the posterior distribution of the FOSR-AR model with 
$\alpha_{j,k,t}=  \alpha_{j,k}$, set $\sigma_{\epsilon_t} = \sigma_\epsilon$ with a Jeffreys' prior $\left[\sigma_\epsilon^2 \right]\propto 1/\sigma_\epsilon^2$. Application of the FOSR-AR model requires an exchangeability assumption: we assume the regression effects $\alpha_{j,k}$ are common across nations, and allow for the regression errors $\gamma_{k,t}$ to be autocorrelated in time $t$, even when different times $t$ correspond to different nations. The time-varying parameter DFOSR produced similar results (the simulations of Section \ref{simulations} suggest that, even when the true model is a DFOSR, the non-dynamic parameter model FOSR-AR may be preferable for small sample sizes $T \le 50$). We report results for $K=3$; larger values of $K$ produce nearly identical results. The MCMC is efficient:  the computation time for 25000 iterations of the Gibbs sampling algorithm (with $T = 20$, $M = 31$, and $p = 6$), implemented in \texttt{R} (on a MacBook Pro, 2.7 GHz Intel Core i5), is less than 3 minutes. We discard the first 10000 simulations as a burn-in and retain every 3rd sample. Traceplots indicate good mixing and suggest convergence (see the Appendix).

In Figure \ref{fig:fitted_flc_fertility}, we plot the ASFRs with the model-imputed ASFR curves $\hat Y_t(\bm\tau) = \sum_{k=1}^K f_k(\bm\tau) \beta_{k,t}$ and the loading curves $f_k(\bm\tau)$ for $\bm \tau = 17,\ldots,47$ with 95\% simultaneous credible bands \citep{ruppert2003semiparametric}. The fitted ASFR curves $\hat Y_t$ demonstrate an overall decrease in the fertility rate from 2000 to 2016, but this effect is not uniform: the largest decrease occurs for ages 27-37, while the fertility for ages less than 20 actually increased. Importantly, the 95\% simultaneous credible bands for $\hat Y_t$ do not overlap, which confirms that these ASFR curves have indeed changed over time. The loading curves are smooth and describe the dominant modes of variability in the ASFRs. Much of the variability in the $\{f_k\}$ occurs between the ages of 20-40, which further supports the use of \emph{age-specific}, rather than total, fertility rates.

\begin{figure}[h]
\begin{center}
\includegraphics[width=.49\textwidth]{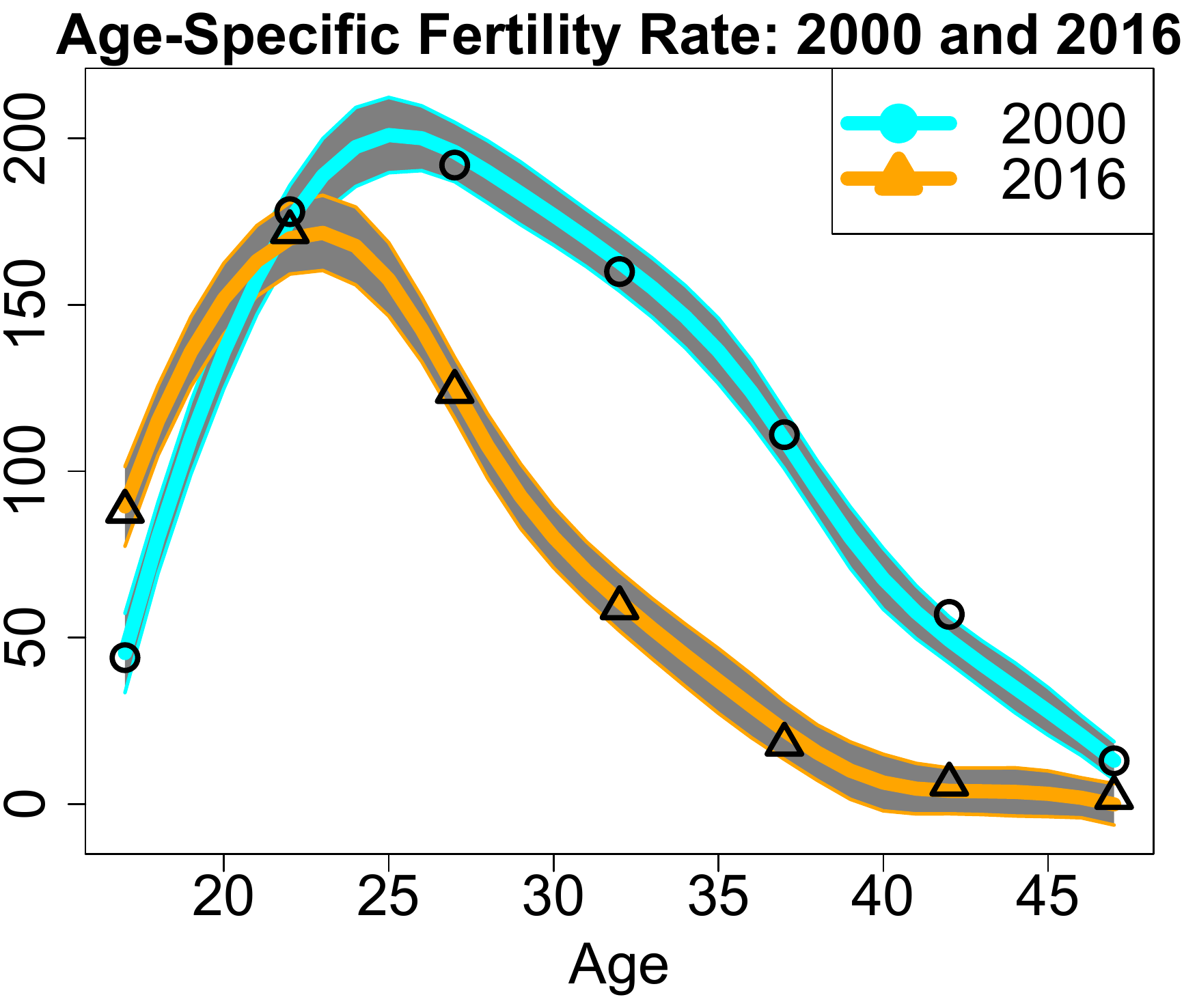}
\includegraphics[width=.49\textwidth]{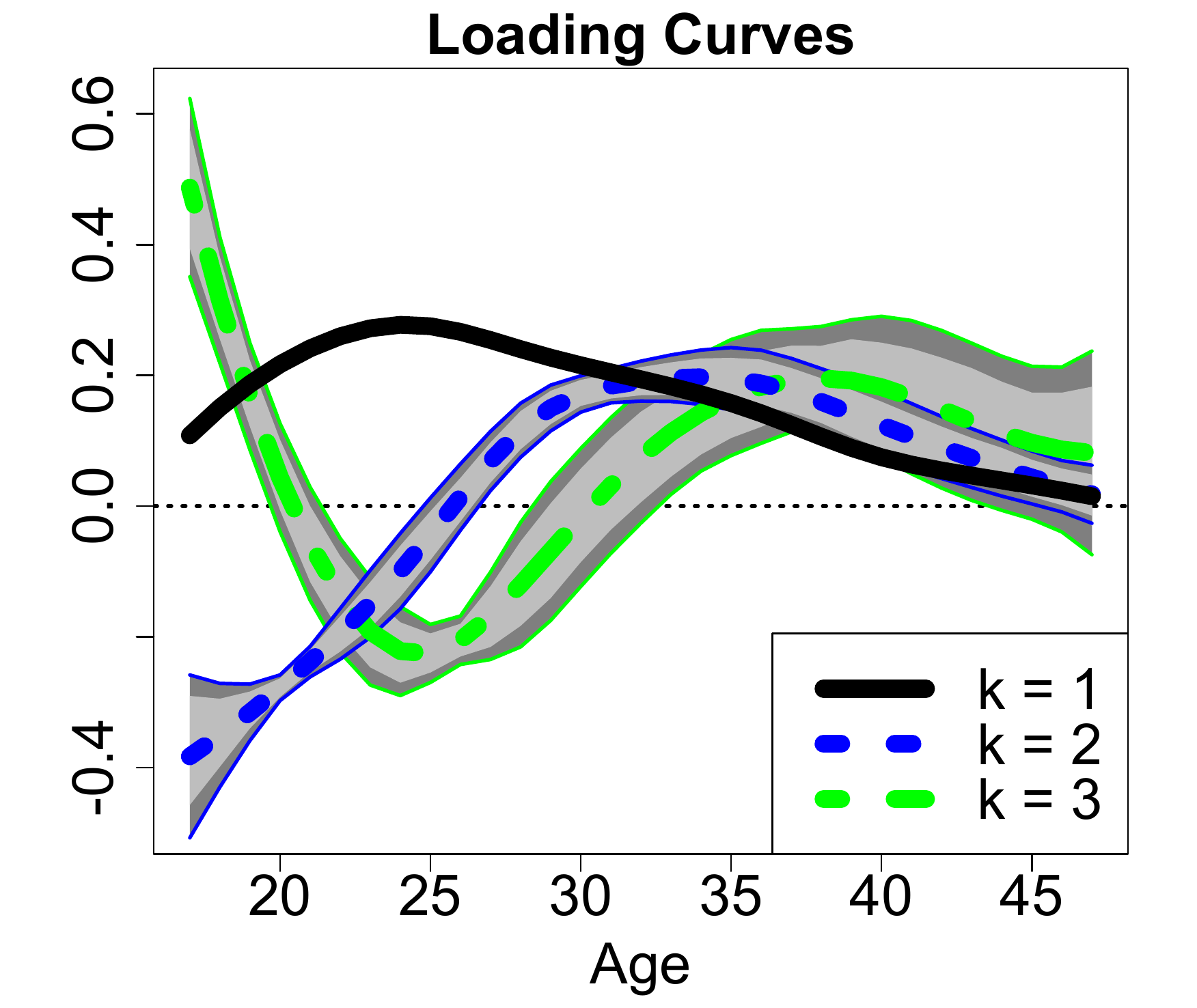}
\caption{{\bf (Left)} Age-specific fertility rates for South and Southeast Asia in 2000 and 2016. For each year $t$, the solid lines are the posterior means of $\hat Y_t(\bm\tau) = \sum_{k=1}^K f_k(\bm\tau) \beta_{k,t}$ and the gray bands are 95\% simultaneous credible bands for  $\hat Y_t(\bm\tau)$, where $\bm \tau =17,\ldots,47$ years of age. {\bf (Right)} Estimated loading curves $f_k$. For each curve $f_k(\bm\tau)$, the solid line is the posterior mean, the light gray bands are 95\% pointwise credible intervals, and the dark gray bands are 95\% simultaneous credible bands.
 \label{fig:fitted_flc_fertility}}
\end{center}
\end{figure}

In Figure \ref{fig:coef_fertility}, we plot the (static) regression functions $\tilde\alpha_j(\bm\tau) = \sum_{k=1}^K f_k(\bm\tau) \alpha_{j,k}$ for each predictor $j=1,\ldots,p$, which may be interpreted via model \eqref{dfosr_fts}. The 95\% simultaneous credible bands exclude zero for both (i) the percentage of currently married or in union women currently using any method of contraception and  (ii) the median age of first marriage or union in years among women (age 25-49), which indicates that these variables are important for ASFRs.  The U-shaped coefficient function in Figure  \ref{fig:coef_fertility} suggests that a greater percentage of married women with access to contraceptives corresponds to a decline in the expected fertility rate, specifically among women aged 22-45. The S-shaped coefficient function in Figure  \ref{fig:coef_fertility} suggests that a larger median age of first marriage corresponds to a \emph{decrease} in the expected fertility rate among women aged 17-23 and an \emph{increase} in the expected fertility rate among women aged 30-40. Importantly, these results are \emph{age-specific}: the association between each predictor and the fertility rate varies by age, while the smoothness of loading curves $f_k$ implies that similar ages should have similar associations.

\begin{figure}[h]
\begin{center}
\includegraphics[width=.24\textwidth]{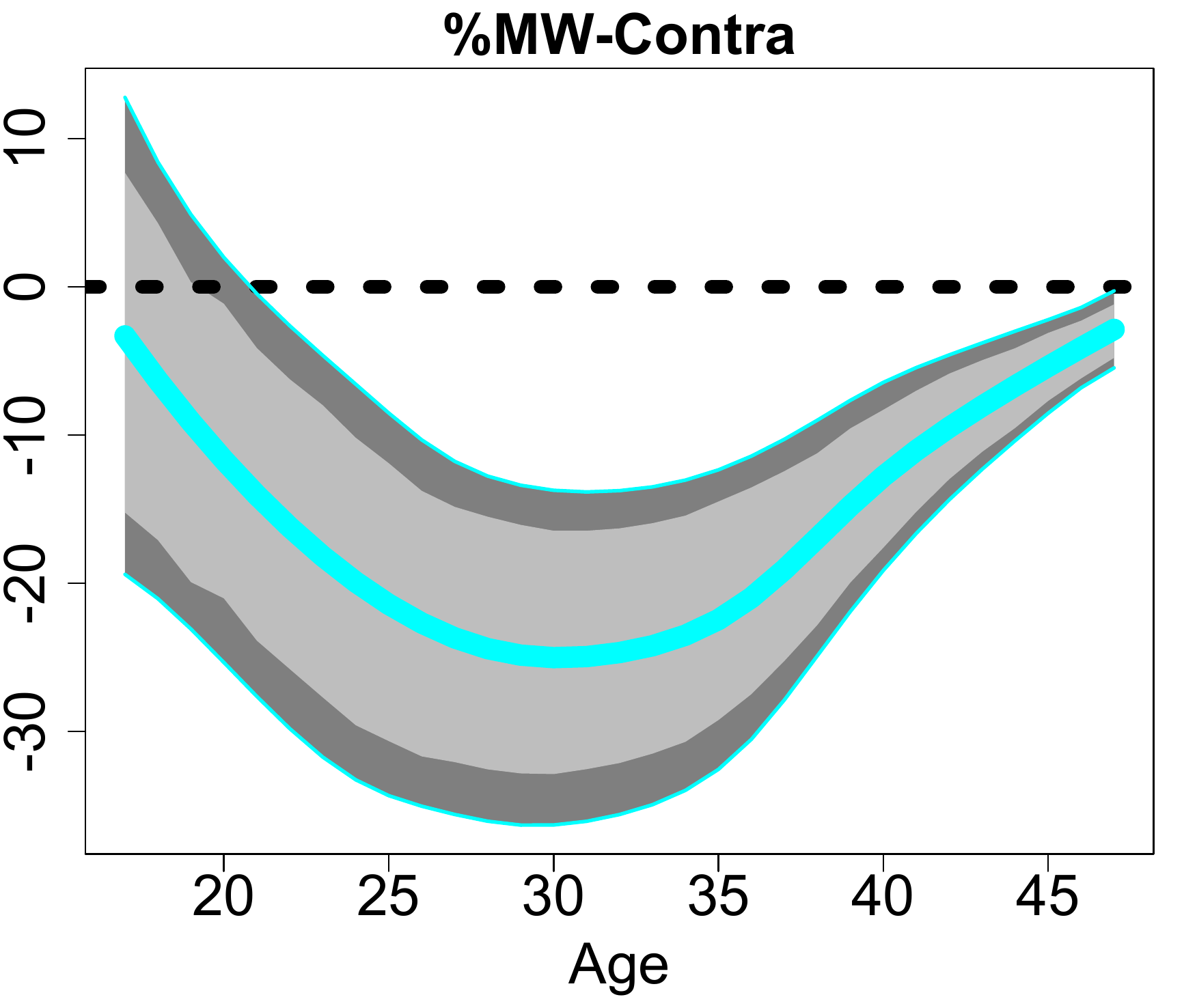}
\includegraphics[width=.24\textwidth]{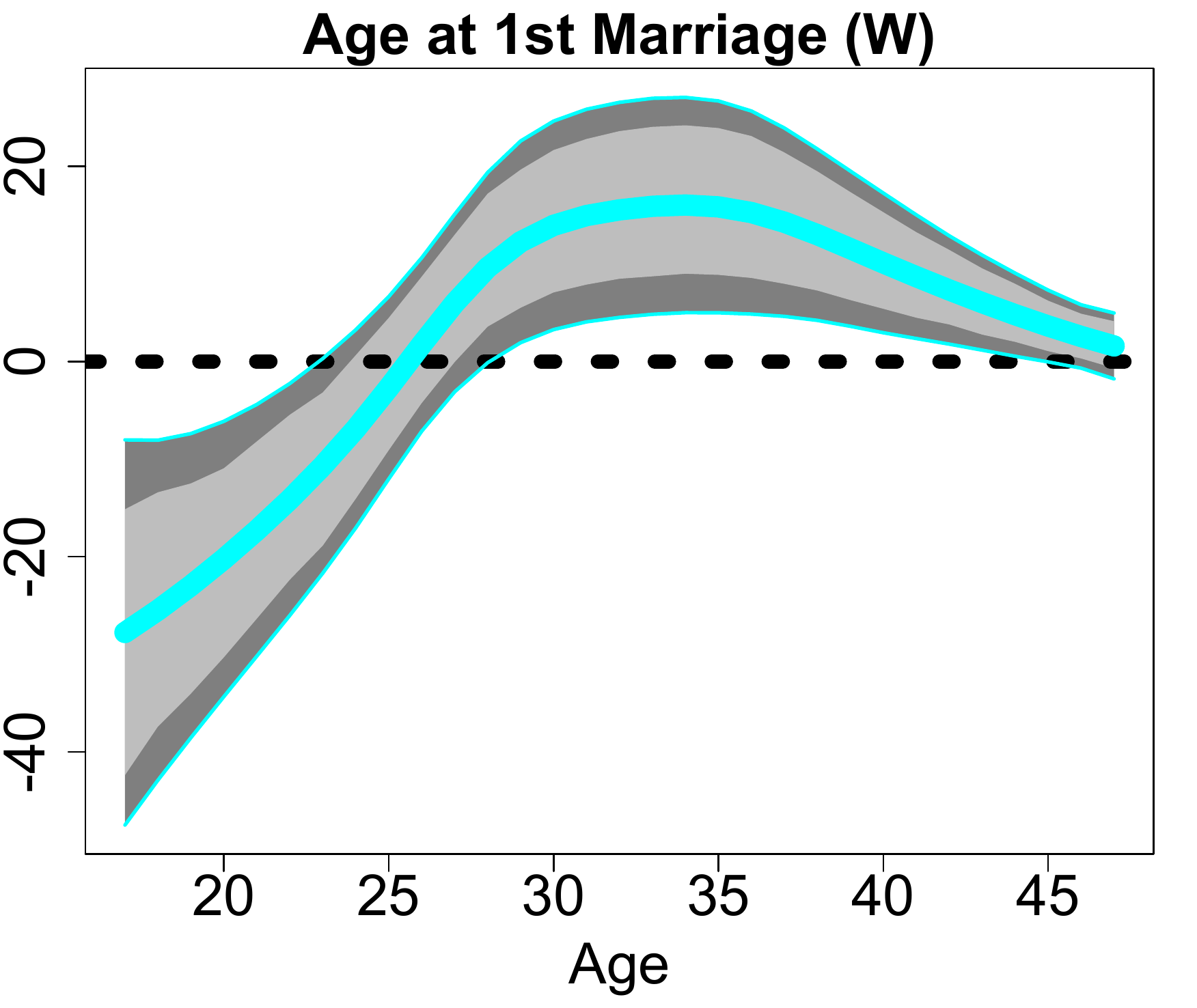}
\includegraphics[width=.24\textwidth]{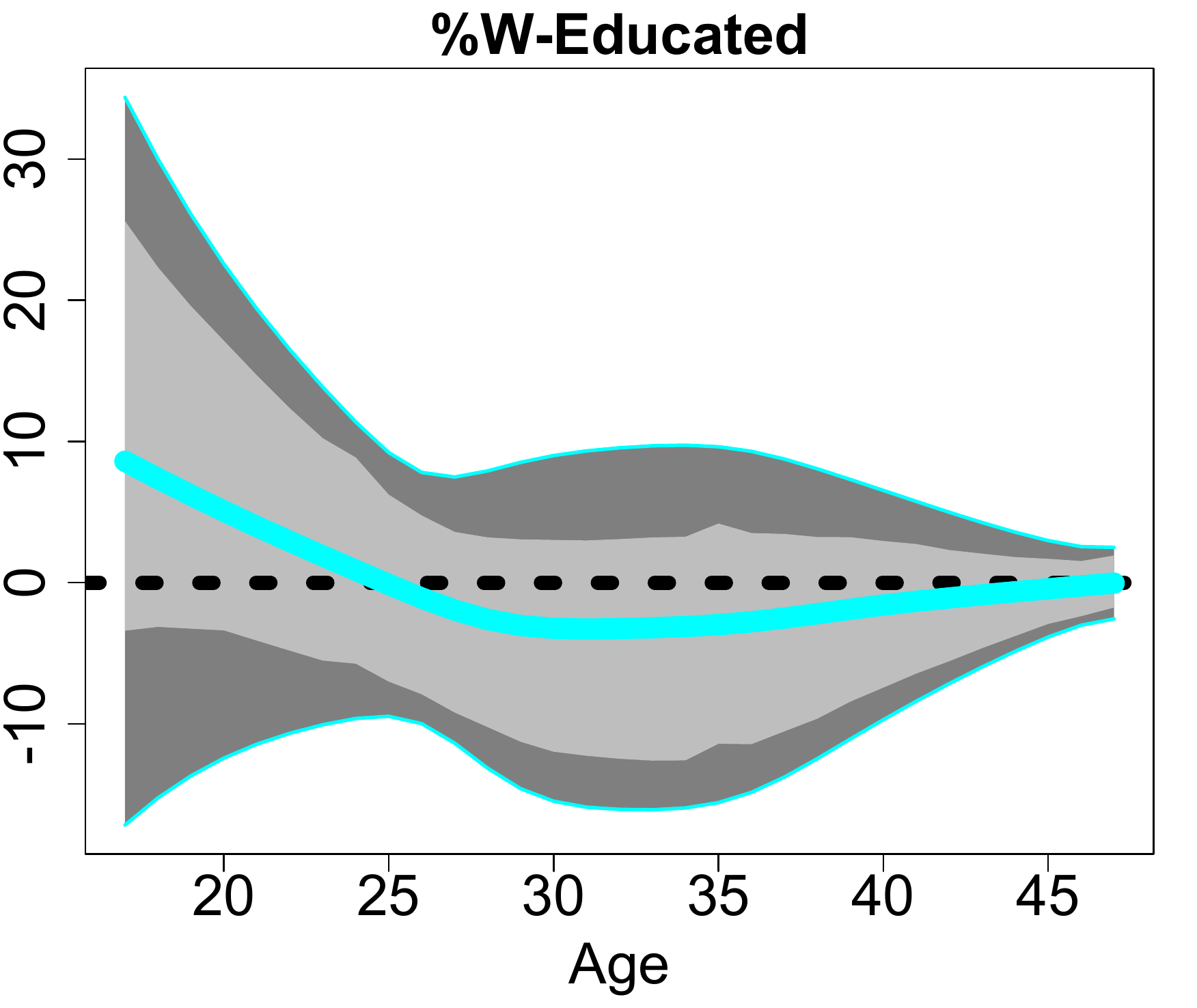}
\includegraphics[width=.24\textwidth]{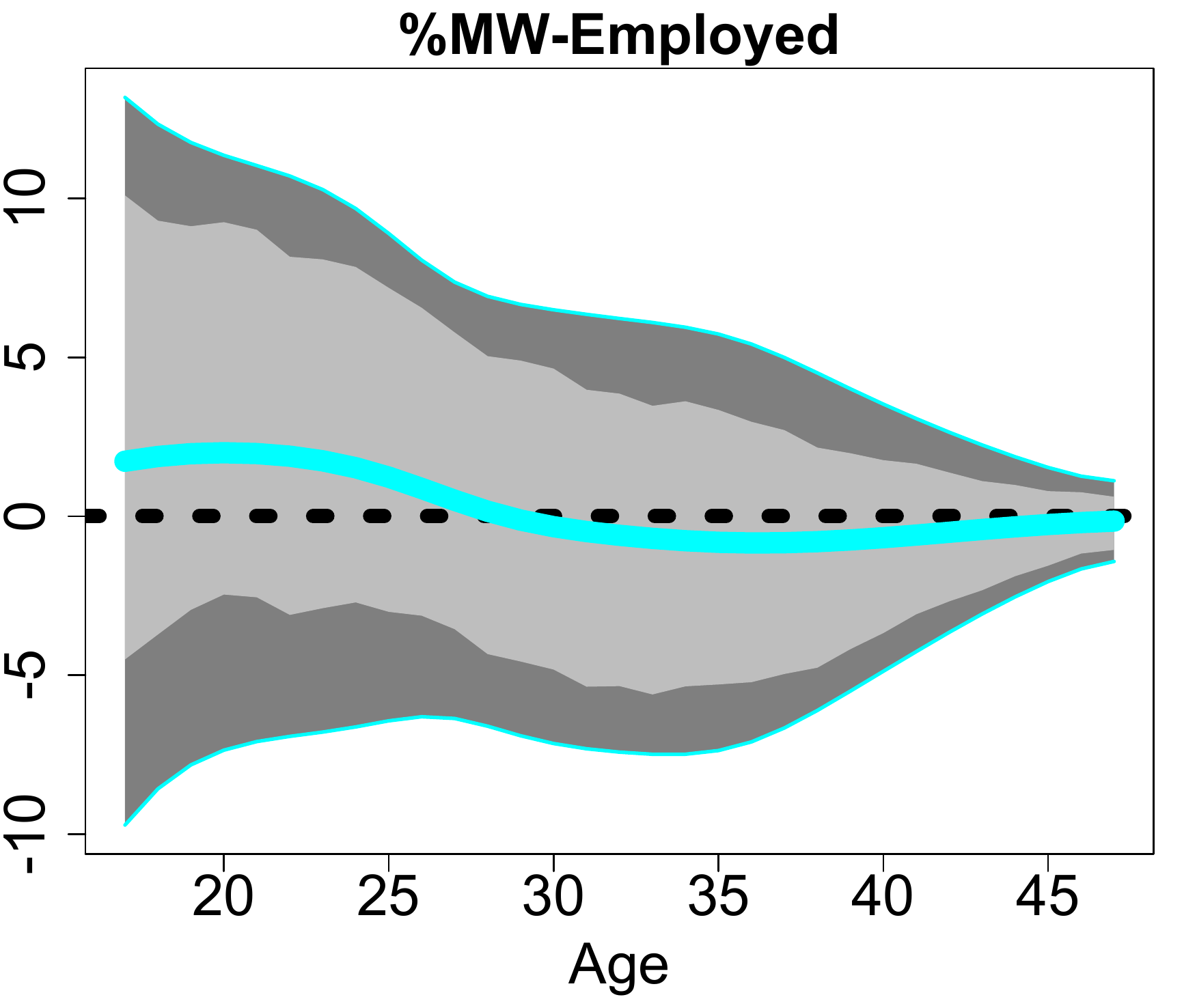}
\caption{Estimated regression function for the percentage of married women using contraceptives {\bf (left)}, the median age of first marriage among women {\bf (left center)}, the percentage of women with secondary or higher education {\bf (right center)}, and the percentage of married women employed in the 12 months preceding the survey {\bf (right)}. For each (static) regression function $\tilde\alpha_j(\bm\tau) = \sum_{k=1}^K f_k(\bm\tau) \alpha_{j,k}$, the solid line is the posterior mean, the light gray bands are 95\% pointwise credible intervals, and the dark gray bands are 95\% simultaneous credible bands.
 \label{fig:coef_fertility}}
\end{center}
\end{figure}

\section{MCMC Sampling Algorithm} \label{MCMC}
We develop an efficient Gibbs sampling algorithm for model \eqref{fts}-\eqref{evol} based on four essential components: (i) the loading curve sampler for $\{f_k\}$ with the identifiability constraint $\bm F' \bm F = \bm I_K$; (ii) the projection-based simplification of the likelihood \eqref{ftsVec} from Lemma \ref{ftsLike2}; (iii) a state space simulation smoother for the dynamic regression parameters in \eqref{reg} and \eqref{evol}; and (iv) parameter expansions  for the variance components in \eqref{fts}, \eqref{reg}, and \eqref{evol}. For sparsely observed functional data, in which the functional data $Y_t$ are \emph{not} observed at the same observation points $\bm \tau_1,\ldots,\bm \tau_M$ for all times $t$, we include a sampling-based imputation step as in Section \ref{fertility}. Since components (i) and (ii) are discussed in Section \ref{loadings} and component (iv) uses standard techniques for Bayesian shrinkage, we focus on (iii) here. The details of the full Gibbs sampling algorithm are provided in the Appendix.

Using Lemma \ref{ftsLike2}, we project the functional data $\bm Y_t$ on the loading curves $f_k$ to obtain the \emph{working likelihood} \eqref{ftsVec2}. Combining the dynamic terms from \eqref{reg}-\eqref{evol} into \emph{state variables} with likelihood \eqref{ftsVec2}, we have
\begin{align}
\label{fts2}
 \tilde Y_{k,t} &= \mu_k + \begin{pmatrix}  \bm x_t' & 1\end{pmatrix} \begin{pmatrix} \bm \alpha_{k,t} \\ \gamma_{k,t}  \end{pmatrix}  + \tilde \epsilon_{k,t}\\
\label{regEvol2}
\begin{pmatrix} \bm \alpha_{k,t} \\ \gamma_{k,t}  \end{pmatrix}  &= \begin{pmatrix} \bm I_p & 0 \\ 0 & \phi_k \end{pmatrix} \begin{pmatrix} \bm \alpha_{k,t-1} \\ \gamma_{k,t-1}  \end{pmatrix} + \begin{pmatrix} \bm \omega_{k,t} \\ \eta_{k,t}  \end{pmatrix}  
\end{align}
where $\bm \alpha_{k,t} = (\alpha_{1,k,t},\ldots,\alpha_{p,k,t})'$ and the errors $\tilde \epsilon_{k,t}$ and $(\bm \omega_{k,t}', \eta_{k,t})'$ are mutually independent and conditionally Gaussian. 
The resulting model is a \emph{dynamic linear model} \citep{westDLM} in the state variables $(\bm \alpha_{k,t}', \gamma_{k,t})'$, and therefore the parameters $\{\bm \alpha_{k,t}, \gamma_{k,t}\}_{t=1}^T$ may be sampled \emph{jointly} across all $t=1,\ldots,T$ using efficient state space simulation methods  \citep{durbin2002simple}. These samplers are also valid for FOSR-AR with 
$\alpha_{j,k,t}=  \alpha_{j,k}$. Note that the model \eqref{fts2}-\eqref{regEvol2} may be aggregated across $k=1,\ldots,K$ to produce a jointly sampler with respect to $k$; in our experience, however, doing so increases computation time without improving MCMC efficiency. A single draw of all dynamic regression coefficients and autoregressive regression error terms   $\{\bm \alpha_{k,t}, \gamma_{k,t}\}_{k,t}$ \emph{jointly} has computational complexity $\mathcal{O}(KTp^3)$. For small to moderate number of predictors $p < 30$, the algorithm is  efficient; for sufficiently small $K$, the sampler is nearly computationally equivalent to the analogous \emph{non-functional} time-varying parameter regression model. 

In addition to the loading curve sampler for $\{f_k\}$ in Section \ref{loadings} and the state space simulation sampler for $\{\bm \alpha_{k,t}, \gamma_{k,t}\}_{k,t}$ via \eqref{fts2}-\eqref{regEvol2}, the Gibbs sampler proceeds by iteratively sampling the intercepts $\{\mu_k\}$, the autoregressive coefficients $\{\phi_k\}$, and the variance components $\sigma_{\epsilon_t}^2$, $\sigma_{\eta_{k,t}}^2$, and $\sigma_{\omega_{j,k,t}}^2$---as well as any relevant hyperparameters---from their full conditional distributions (see the Appendix). Posterior inference is therefore available for these quantities as well as the time-varying parameter regression functions $\tilde\alpha_j(\bm\tau) = \sum_{k=1}^K f_k(\bm\tau) \alpha_{j,k}$ from Proposition \ref{model_equiv} and the fitted curves $\hat Y_t(\bm\tau) = \sum_{k=1}^K f_k(\bm\tau) \beta_{k,t}$ with $\beta_{k,t}$ defined in \eqref{reg}. 

\section{Discussion and Future Work}\label{conclusions}
The proposed \emph{dynamic function-on-scalars regression} model provides a fully Bayesian framework for simultaneously modeling functional dependence,  time dependence, and dynamic predictors. We incorporate a nonparametric model for functional dependence, an autoregressive model for time-dependence, and a time-varying parameter regression model for dynamic predictors. The model is flexible, yet incorporates appropriate shrinkage and smoothness priors to guard against overfitting. A simulation study validates our model for the loading curves $f_k$ (Section \ref{loadings}) and our choice of shrinkage priors (Section \ref{shrinkage}) by demonstrating substantial improvements in estimation accuracy relative to existing methods as well as simpler submodels. Applications in yield curves and age-specific fertility rates illustrate the utility of our approach: in particular, we provide estimation, uncertainty quantification, and imputation for regression coefficient \emph{functions}, which may be \emph{time-varying}. 

Future work will extend model \eqref{fts} for other important dependence structures, such as dynamic \emph{functional} predictors $X_{j,t}(\bm u)$ for $\bm u \in \mathcal{U}$, possibly with different domains $\mathcal{U} \ne \mathcal{T}$. Notably, our  efficient projection-based Gibbs sampler only requires the likelihood \eqref{fts} and the identifiability constraint $\bm F' \bm F = \bm I_K$ to obtain the working likelihood \eqref{ftsVec2}. Therefore, it is straightforward to combine our nonparametric model for the loading curves $f_k$ with alternative models for $\beta_{k,t}$ in \eqref{reg}-\eqref{evol}, while maintaining computational scalability. 

\bibliographystyle{apalike}
\bibliography{BFDLMbib}

\appendix
\section{Appendix}

\subsection*{MCMC Algorithm} \label{extra:full}
The \emph{dynamic function-on-scalars regression model} (DFOSR), with all prior distributions, is 
\begin{align}
\label{extra:fts}
Y_t(\bm \tau) &= \sum_{k=1}^K f_k(\bm \tau) \beta_{k,t} + \epsilon_t(\bm \tau), \quad \epsilon_t(\bm\tau) \stackrel{indep}{\sim}N(0, \sigma_{\epsilon}^2), \quad  \sigma_\epsilon^2 \propto 1/\sigma_\epsilon^2  \\
\label{extra:flcs} 
f_k(\bm\tau) &= \bm b'(\bm\tau) \bm \psi_k, \quad \bm \psi_k \stackrel{indep}{\sim} N\left(\bm 0, \lambda_{f_k}^{-1} \bm \Omega^{-1}\right), \quad \lambda_{f_k}^{-1/2} \stackrel{iid}{\sim} \mbox{Uniform}(0, 10^4)\\
\label{extra:reg}
\beta_{k,t} &= \mu_{k} +  \sum_{j=1}^p x_{j,t} \alpha_{j,k,t} + \gamma_{k,t}, \quad \gamma_{k,t} = \phi_k \gamma_{k,t-1} + \eta_{k,t}, \quad \eta_{k,t} \stackrel{indep}{\sim}N(0, \sigma_{\eta_{k,t}}^2) \\
\label{extra:mu}
\mu_k &\stackrel{indep}{\sim} N(0, \sigma_{\mu_k}^2), \quad \left[(\phi_k + 1)/2\right]\stackrel{iid}{\sim} \mbox{Beta}(5, 2)\\
\label{extra:muMGP}
\sigma_{\mu_k}^{-2} &= \prod_{\ell \le k} \delta_{\mu_\ell}, \quad \delta_{\mu_1} \sim \mbox{Gamma}(a_{\mu_1}, 1), \quad \delta_{\mu_\ell} \sim \mbox{Gamma}(a_{\mu_2}, 1), \quad \ell > 1 \\
 \label{extra:eta}
\sigma_{\eta_{k,t}}^2 &= \sigma_{\eta_k}^2/\xi_{\eta_{kt}}, \quad \xi_{\eta_{kt}} \stackrel{iid}{\sim} \mbox{Gamma}(\nu_\eta/2, \nu_\eta/2), \quad \nu_\eta \sim \mbox{Unif}(2, 128)
\\
 \label{extra:etaMGP}
 \sigma_{\eta_k}^{-2} &= \prod_{\ell \le k} \delta_{\eta_\ell}, \quad \delta_{\eta_1} \sim \mbox{Gamma}(a_{\eta_1}, 1), \quad \delta_{\eta_\ell} \sim \mbox{Gamma}(a_{\eta_2}, 1), \quad \ell > 1 \\
  \label{extra:hyper}
&  a_{\mu_1}, a_{\mu_2}, a_{\eta_1},a_{\eta_2} \stackrel{iid}{\sim}\mbox{Gamma}(2,1)\\
\label{extra:evol}
\alpha_{j,k,t} &= \alpha_{j,k,t-1}  + \omega_{j,k,t}, \quad \omega_{j,k,t} \stackrel{indep}{\sim}N(0, \sigma_{\omega_{j,k,t}}^2)  \\
\label{extra:hs}
\sigma_{\omega_{j,k,t}} &\stackrel{ind}{\sim} C^+(0, \lambda_{j,k}), \quad \lambda_{j,k} \stackrel{ind}{\sim} C^+(0, \lambda_{j}), \quad \lambda_{j} \stackrel{ind}{\sim} C^+(0, \lambda_{0}), \quad \lambda_{0} \stackrel{ind}{\sim} C^+(0, 1/\sqrt{T-1}) \\
\label{extra:inits}
\eta_{k,0} &\stackrel{iid}{\sim} t_3(0,1), \quad \omega_{j,k,0} \stackrel{iid}{\sim} t_3(0,1)
\end{align}
for  $\bm \tau \in \mathcal{T}$, $j=1,\ldots,p$, $k=1,\ldots,K$, and $t=1,\ldots,T$. The details for each level are described in the main paper. Note that $\bm \Omega$ in \eqref{extra:flcs} may not be invertible, but for low-rank thin plate splines the posterior distribution of $\bm \psi_k$ will be proper. In the yield curve application of Section \ref{yields}, the Jeffreys prior in \eqref{extra:fts} is replaced by a stochastic volatility model for the variance $\sigma_{\epsilon_t}^2$. Specifically, the model is an AR(1) for the log-variance $h_t = \log \sigma_{\epsilon_t}^2$: $h_{t+1} = \mu_h + \phi_h(h_t - \mu_h) + \nu_{h_t}$, where $\mu_h \sim N(-10, 100)$ is the unconditional mean of log-volatility, $\left[(\phi_h + 1)/2\right] \sim \mbox{Beta}(20, 1.5)$ is the autoregressive parameter, and $\nu_{h_t} \stackrel{iid}{\sim}N(0, \sigma_{\nu_h}^2)$ is the log-volatility innovation with standard deviation $\sigma_{\nu_h} \sim  \mbox{Uniform}(0, 100)$. Sampling $\{h_t\}$ is a straightforward modification of the algorithm in \cite{kastner2014ancillarity}, and conditional on $\{h_t\}$, the parameters $\mu_h, \phi_h$, and $\sigma_{\nu_h}$ may be sampled iteratively using standard procedures for Bayesian autoregressive models.

We construct a Gibbs sampling algorithm that primarily features draws from known full conditional distributions with a small number of slice sampling steps \citep{neal2003slice}. For the half-Cauchy and t-distributions in \eqref{extra:hs} and \eqref{extra:inits}, respectively, we use the following scale mixture of Gaussian parameter expansions. The hierarchy of half-Cauchy distributions may be written on the precision scale with Gamma expansions:
$[\sigma_{\omega_{j,k,t}}^{-2} | \xi_{\sigma_{\omega_{j,k,t}}}] \sim \mbox{Gamma}(1/2, \xi_{\sigma_{\omega_{j,k,t}}})$, $[\xi_{\sigma_{\omega_{j,k,t}}} | \lambda_{j,k}] \sim \mbox{Gamma}(1/2, \lambda_{j,k}^{-2})$, $[\lambda_{j,k}^{-2} | \xi_{\lambda_{j,k}}] \sim \mbox{Gamma}(1/2, \xi_{\lambda_{j,k}})$, $[\xi_{\lambda_{j,k}} | \lambda_{j}] \sim \mbox{Gamma}(1/2, \lambda_{j}^{-2})$, $[\lambda_{j}^{-2} | \xi_{\lambda_{j}}] \sim \mbox{Gamma}(1/2, \xi_{\lambda_{j}})$, $[\xi_{\lambda_{j}} | \lambda_{0}] \sim \mbox{Gamma}(1/2, \lambda_{0}^{-2})$, $[\lambda_{0}^{-2} | \xi_{\lambda_{0}}] \sim \mbox{Gamma}(1/2, \xi_{\lambda_{0}})$, and $[\xi_{\lambda_{0}}] \sim \mbox{Gamma}(1/2, T-1)$. The t-distributions are expanded as 
$[\eta_{k,0}| \xi_{\eta_{k,0}}] \sim N(0, 1/\xi_{\eta_{k,0}})$ and $\xi_{\eta_{k,0}} \sim \mbox{Gamma}(3/2, 3/2)$ and similarly, 
$[\omega_{j,k,0}| \xi_{\omega_{j,k,0}}] \sim N(0, 1/\xi_{\omega_{j,k,0}})$ and $\xi_{\omega_{j,k,0}} \sim \mbox{Gamma}(3/2, 3/2)$. In all cases, the full conditional distributions are Gamma (on the precision scale). 

\subsection*{Gibbs Sampling Algorithm}
\begin{enumerate}
\item {\bf Imputation:} for all unobserved $Y_t(\bm \tau_t^*)$, sample each $\left[Y_t(\bm \tau_t^*) | \{f_k\}, \{\beta_{k,t}\}, \{\sigma_{\epsilon_t}\}\right] \stackrel{indep}{\sim} N\big(\sum_k f_k(\bm \tau_t^*) \beta_{k,t}, \sigma_{\epsilon_t}^2\big)$.
\item {\bf Loading curves and smoothing parameters:} for $k=1,\ldots,K$, 
\begin{enumerate}
\item Sample $[\lambda_{f_k} | \cdots] \sim \mbox{Gamma}((L_M - D + 1 + 1)/2, \bm\psi_k' \bm\Omega \bm\psi_k/2)$ truncated to $(10^{-8}, \infty)$.
\item Sample $\left[\bm \psi_k | \cdots \right] \sim N\left(\bm Q_{\psi_k}^{-1} \bm \ell_{\psi_k}, \bm Q_{\psi_k}^{-1}\right)$ conditional on $\bm C_k \bm \psi_k = \bm 0$, where $\bm C_k = (\bm f_1, \ldots, \bm f_{k-1}, \bm f_{k+1}, \ldots, \bm f_K)' \bm B= (\bm \psi_1, \ldots, \bm \psi_{k-1}, \bm \psi_{k+1}, \ldots, \bm \psi_K)'$, using a 
modified version of the efficient Cholesky decomposition approach of \cite{wand2008semiparametric}:
\begin{enumerate}
\item Compute the (lower triangular) Cholesky decomposition $\bm Q_{\psi_k}   = \bar{\bm{Q}}_L\bar{\bm{Q}}_L'$;
\item Use forward substitution to obtain $\bar{\bm{\ell}}$ as the solution to $\bar{\bm{Q}}_L \bar{\bm{\ell}} =\bm \ell_{\psi_k}$,  then use backward substitution to obtain $\bm \psi_k^0$ as the solution to $\bar{\bm{Q}}_L' \bm \psi_k^0 = \bar{\bm{\ell}} + {\bm{z}}$, where ${\bm{z}}\sim N (\mathbf{0},  \mathbf{I}_{L_M})$;
\item Use forward substitution to obtain $\utwi{\bar C}$ as the solution to $\bar{\bm{Q}}_L \utwi{\bar C}= \bm C_k $,  then use backward substitution to obtain $\utwi{\tilde C}$ as  the solution to $\bar{\bm{Q}}_L'\utwi{\tilde C} = \utwi{\bar C}$;
\item Set $\bm \psi_k^* = \bm \psi_k^0 - \utwi{\tilde C} (\bm C_k\utwi{\tilde C})^{-1}\bm C_k \bm \psi_k^0$;
\item Retain the vectors $\bm \psi_k = \bm \psi_k^*/\sqrt{{\bm \psi_k^*}'\bm B' \bm B\bm \psi_k^*} = \bm \psi_k^*/ || \bm \psi_k^*||$  and $\bm f_k = \bm B \bm \psi_k$ and update $\beta_{k,t} \leftarrow \beta_{k,t} || \bm \psi_k^*||$.
\end{enumerate} 
\end{enumerate}
\item {\bf Project:} update $ \tilde Y_{k,t} = \bm f_k' \bm Y_t = \bm \psi_k' \left(\bm B' \bm Y_t\right)$ for all $k,t$.
\item {\bf Dynamic state variables:} sample $[\{\alpha_{j,k,t}\}, \{\gamma_{k,t}\} | \{\tilde Y_{k,t}\}, \cdots]$ jointly, including the initial states $\{\eta_{k,0}\}$ and $\{\omega_{j,k,0}\}$, using \cite{durbin2002simple}. \\
\emph{Note:} we condition on $\{\mu_k\}$ for computational efficiency (i.e., a smaller state vector), but $\mu_k$ could be included in this joint sampler.
\item {\bf Unconditional mean and AR coefficients:} for $k=1,\ldots,K$,
\begin{enumerate}
\item Using the \emph{centered} AR parametrization with $\gamma_{k,t}^c = \gamma_{k,t} + \mu_k$ (computed with the previous simulated value of $\mu_k$), so $\gamma_{k,t}^c = \mu_k  + \phi_k(\gamma_{k,t-1}^c - \mu_k) + \eta_{k,t}$, sample 
$[\mu_k | \cdots] \stackrel{indep}{\sim} N(Q_{\mu_k}^{-1}, \ell_{\mu_k}, Q_{\mu_k}^{-1})$ where $Q_{\mu_k} = \sigma_{\mu_k}^{-2} + (1-\phi_k)^2 \sum_{t=2}^T \sigma_{\eta_{k,t}}^{-2}$ and $\ell_{\mu_k} = (1-\phi_k) \sum_{t=2}^T (\gamma_{k,t}^c - \phi_k \gamma_{k,t-1}^c)\sigma_{\eta_{k,t}}^{-2}$.
\item Sample $\phi_k$ using the slice sampler \citep{neal2003slice}.
\end{enumerate}
\item {\bf Variance parameters:}
\begin{enumerate}
\item {\bf Observation error variance:} $[\sigma_\epsilon^{-2} | \cdots] \sim \mbox{Gamma}\left(\frac{MT}{2}, \frac{1}{2}\sum_{t=1}^T ||\bm Y_t - \bm F \bm\beta_t ||^2\right)$

\item {\bf Multiplicative Gamma Process Parameters:} given $\mu_k$ and $\eta_{k,t} = \gamma_{k,t} - \phi_k \gamma_{k,t-1}$ for $\gamma_{k,t} = \gamma_{k,t}^c - \mu_k$ (after sampling $\mu_k$ above),
\begin{enumerate}
\item Sample $[\delta_{\mu_1} | \cdots ] \sim \mbox{Gamma}\big(a_{\mu_1} + \frac{K}{2}, 1 + \frac{1}{2} \sum_{k=1}^K \tau_{\mu_k}^{(1)} \mu_k^2 \big)$ and $[\delta_{\mu_\ell} | \cdots ] \sim \mbox{Gamma}\big(a_{\mu_2} + \frac{K - \ell + 1}{2}, 1 + \frac{1}{2} \sum_{k=\ell}^K \tau_{\mu_k}^{(\ell)} \mu_k^2 \big)$ for $\ell > 1$ where $\tau_{\mu_\ell}^{(k)} = \prod_{h = 1, h \ne k}^\ell \delta_{\mu_h}$.
\item Set $\sigma_{\mu_k} = \prod_{\ell \le k} \delta_{\mu_\ell}^{-1/2}$.
\item Sample $[\delta_{\eta_1} | \cdots ] \sim \mbox{Gamma}\big(a_{\eta_1} + \frac{K(T-1)}{2}, 1 + \frac{1}{2} \sum_{k=1}^K \tau_{\eta_k}^{(1)} \sum_{t=2}^T \eta_{k,t}^2 \xi_{\eta_{k,t}} \big)$ and $[\delta_{\eta_\ell} | \cdots ] \sim \mbox{Gamma}\big(a_{\eta_2} + \frac{(K - \ell + 1)(T-1)}{2}, 1 + \frac{1}{2} \sum_{k=\ell}^K \tau_{\eta_k}^{(\ell)} \sum_{t=2}^T \eta_{k,t}^2 \xi_{\eta_{k,t}} \big)$ for $\ell > 1$ where $\tau_{\eta_\ell}^{(k)} = \prod_{h = 1, h \ne k}^\ell \delta_{\eta_h}$.
\item Set $\sigma_{\eta_k} = \prod_{\ell \le k} \delta_{\eta_\ell}^{-1/2}$ 
\item Sample $[\xi_{\eta_{k,t}} | \cdots] \stackrel{indep}{\sim} \mbox{Gamma}\big(
\frac{\nu_\eta}{2} + \frac{1}{2}, \frac{\nu_\eta}{2} + \frac{\eta_{k,t}^2}{2\sigma_{\eta_k}^2}
\big)$
\item Set $\sigma_{\eta_{k,t}} =  \sigma_{\eta_k}/\sqrt{\xi_{\eta_{k,t}}}$.
\end{enumerate}

\item {\bf Hierarchical Half-Cauchy Parameters:} for $\omega_{j,k,t} = \alpha_{j,k,t} - \alpha_{j,k,t-1}$, 
\begin{enumerate}
\item Sample $[\sigma_{\omega_{j,k,t}}^{-2} | \cdots ] \stackrel{indep}{\sim}  \mbox{Gamma}\big(1, \xi_{\sigma_{\omega_{j,k,t}}} + \omega_{j,k,t}^2/2\big)$ and  
\\ $[\xi_{\sigma_{\omega_{j,k,t}}} | \cdots] \stackrel{indep}{\sim}  \mbox{Gamma}\big(1, \lambda_{j,k}^{-2} + \sigma_{\omega_{j,k,t}}^{-2}\big)$.

\item Sample $[\lambda_{j,k}^{-2} | \cdots] \stackrel{indep}{\sim}  \mbox{Gamma}\big(\frac{T}{2}, \xi_{\lambda_{j,k}} + \sum_t \xi_{\sigma_{\omega_{j,k,t}}}\big)$ and 
\\ $[\xi_{\lambda_{j,k}} | \cdots] \stackrel{indep}{\sim}  \mbox{Gamma}\big(1, \lambda_{j}^{-2} + \lambda_{j,k}^{-2}\big)$.

\item Sample $[\lambda_{j}^{-2} | \cdots] \stackrel{indep}{\sim}  \mbox{Gamma}\big(\frac{K + 1}{2}, \xi_{\lambda_{j}} + \sum_{k=1}^K \xi_{\lambda_{j,k}} \big)$ and 
\\ $[\xi_{\lambda_{j}} | \cdots] \stackrel{indep}{\sim}  \mbox{Gamma}\big(1, \lambda_{0}^{-2} + \lambda_{j}^{-2}\big)$. 

\item Sample $[\lambda_{0}^{-2} | \cdots] \sim \mbox{Gamma}\big(\frac{p + 1}{2}, \xi_{\lambda_{0}} + \sum_{j=1}^p \xi_{\lambda_{j}} \big)$ 
and 
\\ $[\xi_{\lambda_{0}} | \cdots] \sim \mbox{Gamma}\big(1, (T-1) + \lambda_{0}^{-2} \big)$. 
\end{enumerate}

\item {\bf Parameter-expanded initial values:}
\begin{enumerate}
\item Sample $[\xi_{\eta_{k,0}} | \ldots] \stackrel{indep}{\sim}  \mbox{Gamma}\big(\frac{3}{2} + \frac{1}{2}, \frac{3}{2} + \frac{1}{2} \eta_{k,0}^2\big)$.
\item Sample $[\xi_{\omega_{j,k,0}} | \ldots] \stackrel{indep}{\sim}  \mbox{Gamma}\big(\frac{3}{2} + \frac{1}{2}, \frac{3}{2} + \frac{1}{2} \omega_{j,k,0}^2\big)$.
\end{enumerate}
\end{enumerate}
\item {\bf Hyperparameters:} sample $a_{\mu_1}, a_{\mu_2}, a_{\eta_1},a_{\eta_2}$, and $\nu_\eta$ independently using the slice sampler \citep{neal2003slice}.
\end{enumerate}


\subsection*{Low-Rank Thin Plate Splines}
Thin plate splines are designed  for modeling an unknown smooth function of multiple inputs  $\bm\tau \in\mathcal{T}\subset \mathbb{R}^D$ with $D \in \mathbb{Z}^+$. Thin plate splines place a (known) basis function at every observation point, so $L_M = M$; \emph{low-rank} thin plate splines (LR-TPS) select a smaller set of basis functions  $L_M < M$. LR-TPS can achieve similar estimation accuracy as thin plate splines in a fraction of the computing time, and demonstrate exceptional MCMC efficiency \citep{crainiceanu2005bayesian}. Each LR-TPS $f_k$ has only one hyperparameter $\lambda_{f_k} > 0$, which is a prior precision corresponding to the smoothness parameter \citep{wahba1990spline}.

Given observation points $\bm \tau_j$ for $j=1,\ldots,M$, we construct the basis and penalty matrices in three steps: (i) we build the LR-TPS basis and penalty matrices using the definitions in \cite{wood2006generalized}; (ii) we diagonalize the penalty matrix for an equivalent representation, following \cite{ruppert2003semiparametric} and \cite{crainiceanu2005bayesian}; and (iii) we orthonormalize the basis matrix (and adjust the penalty matrix accordingly). The diagonalization and orthonormalization steps (ii) and (iii) may accompany any choice of basis and penalty matrices, but substantially improve MCMC performance for LR-TPS. Note that while the diagonalization step (ii) is not strictly necessary given the orthonormalization step (iii), it ensures that the final penalty matrix---and therefore the prior precision matrix---is positive definite, which is \emph{not} guaranteed for LR-TPS \citep{ruppert2003semiparametric}.

To build the basis and penalty matrices, we begin by selecting the number and location of knots. For a small number of observation points, $M \le 25$, we use the full rank thin plate spline basis with knots at the unique observation points  $\bm \kappa_\ell = \bm \tau_\ell$. When $M > 25$, we use $(L_M  - D - 1)= \min\{M/4, 150\} $ knots. In the case of $D=1$, knots are selected using the quantiles of the observation points, i.e., $\kappa_\ell$ is the $\left(\ell/L_M \right)$th sample quantile of the unique $\bm\tau_j$; for 
 $D > 1$, we select knot locations using a space-filling algorithm, as in \cite{ruppert2003semiparametric}. Let $\bm W_{0}$ be the $M \times (D+1)$ matrix with $j$th row $[\bm W_{0}]_j = (1, \bm \tau_j')$, $\bm Z_{0}$ be the $M \times (L_M - D - 1)$ matrix with $(j, \ell)$th entry $[\bm Z_0]_{j,\ell}  = b(||\bm \tau_j - \bm \kappa_\ell||)$, and 
 $\bm \Omega_{Z_0}$ be the $(L_M - D - 1) \times (L_M - D - 1)$ penalty matrix with $(\ell, \ell')$th entry $[\bm \Omega_{Z_0}]_{\ell,\ell'}  = b(||  \bm \kappa_\ell - \bm \kappa_{\ell'} ||)$, where $b(r) =r^{4 - D} \log(r)$ for $D$ even and  $b(r) =r^{4 - D}$ for $D$ odd,  $r > 0$, are the (nonlinear) cubic thin plate spline basis functions \citep{wood2006generalized}. The matrices $\bm W_0$ and $\bm Z_0$ constitute the LR-TPS basis matrix, while $\bm \Omega_{Z_0}$ is the LR-TPS penalty matrix. To diagonalize the penalty matrix, let $\bm B_0 = [\bm W_0 : \bm Z_0 \bm \Omega_{Z_0}^{-1/2}]$ be the LR-TPS basis matrix and $\bm \Omega_0 = \mbox{diag}\left(\bm 0_{D+1}', \bm 1_{L_M - D -1}'\right)$ be the diagonalized LR-TPS penalty matrix, where $\bm 0_{D+1}$ is a $(D+1)$-dimensional vector of zeros and $\bm 1_{L_M - D -1}$ is a $(L_M - D -1)$-dimensional vector of ones.  
 Lastly, let $\bm B_0 = \bm Q \bm R$ be the QR decomposition of the initial basis matrix $\bm B_0$, where $\bm Q$ is $L_M \times L_M$ with $\bm Q' \bm Q = \bm I_{L_M}$ and $\bm R$ is $L_M \times L_M$ and upper triangular. Using the orthonormal basis matrix $\bm B = \bm Q$, we reparameterize the penalty matrix $\bm \Omega = (\bm R')^{-1} \bm \Omega_0 \bm R^{-1}$ to obtain an equivalent representation. Notably, this basis matrix $\bm B$ and penalty matrix $\bm \Omega$ construction is a one-time cost. 
 
 For orthonormalized LR-TPS, the full conditional distribution simplifies to $[\bm \psi_k | \cdots ] \sim N\left(\bm Q_{\psi_k}^{-1} \bm \ell_{\psi_k}, \bm Q_{\psi_k}^{-1}\right)$, where $\bm Q_{\psi_k} = \bm I_{L_M} \sum_{t=1}^T \beta_{k,t}^2/\sigma_{\epsilon_t}^2 +  \lambda_{f_k} \bm \Omega$ and $\bm \ell_{\psi_k} =   \sum_{t=1}^T [\beta_{k,t}/\sigma_{\epsilon_t}^2 (\bm B' \bm Y_t)] -  \sum_{t=1}^T [ \beta_{k,t}/\sigma_{\epsilon_t}^2 \sum_{\ell \ne k} \bm \psi_\ell \beta_{\ell,t} ]$. Since $ \lambda_{f_k} > 0$ corresponds to a prior precision parameter, we follow \cite{gelman2006prior} and \cite{kowal2017bayesian} and impose a uniform prior distribution on the corresponding standard deviation,  $\lambda_{f_k}^{-1/2} \stackrel{iid}{\sim} \mbox{Uniform}(0, 10^4)$.


\subsection*{Additional Proofs}\label{extra:proofs}
Given functional data observations $\bm Y_t = (Y_t(\bm \tau_1),\ldots, \bm Y_t(\bm \tau_M))'$ at observation points $\{\bm \tau_j\}_{j=1}^M$, consider the generalization of the likelihood \eqref{ftsVec} from the main paper:
\begin{equation}\label{ftsVecGen}
\bm Y_t = \sum_{k=1}^K \bm f_k \beta_{k,t} + \bm \epsilon_t,
\quad \bm \epsilon_t \stackrel{indep}{\sim} N(\bm 0,\bm \Sigma_{\epsilon_t})
\end{equation}
where  $\bm \Sigma_{\epsilon_t}$ is a general $M\times M$ covariance matrix. 
\begin{lemma}\label{projectionLikeGen}
Under the identifiability constraint $\bm F' \bm F = \bm I_K$, the joint likelihood in \eqref{ftsVecGen} is
\begin{equation}\label{ftsLikeGen}
\begin{aligned}
p\left(\bm Y_1,\ldots, \bm Y_T | \{\bm f_k, \beta_{k,t}, \bm\Sigma_{\epsilon_t}\}_{k,t}\right)
 = c_Y \prod_{t=1}^T \left| \bm \Sigma_{\epsilon_t}\right|^{-1/2} &\exp
 \Big\{
 -\frac{1}{2 } \Big[ 
 \bm Y_t' \bm \Sigma_{\epsilon_t}^{-1} \bm Y_t + \\
 &  \quad \bm \beta_t' \left(\bm F' \bm \Sigma_{\epsilon_t}^{-1}\bm F\right) \bm \beta_t  - 2 \bm \beta_t' \left(\bm F'  \bm \Sigma_{\epsilon_t}^{-1}\bm Y_t\right)
 \Big]
 \Big\}
\end{aligned}
\end{equation}
where $c_Y = (2\pi)^{-MT/2}$ is a constant and $\bm \beta_t' = (\beta_{1,t},\ldots,\beta_{K,t})$. 
\end{lemma}
Analogous to the results in the main paper, Lemma \ref{projectionLikeGen} implies the following  \emph{working likelihood} for the factors $\beta_{k,t}$ and associated parameters:
\begin{lemma}\label{ftsLikeGen2}
Under the identifiability constraint $\bm F' \bm F = \bm I_K$, the joint likelihood in \eqref{ftsVecGen} for $\{\beta_{k,t}\}$ is equivalent to the \emph{working likelihood} implied by
\begin{equation}\label{ftsVecGen2}
\bm{\tilde Y_{t}} =  \bm \beta_{t} + \bm{\tilde \epsilon_{t}},
\quad  \bm{\tilde \epsilon_{t}} \stackrel{indep}{\sim} N(\bm 0, \bm Q_{\beta_t}^{-1})
\end{equation}
up to a constant that does not depend on $\bm \beta_{t}$, where $\bm{\tilde Y_{t}} = \bm Q_{\beta_t}^{-1} \bm \ell_{\beta_t}$ for $\bm Q_{\beta_t} = \bm F' \bm \Sigma_{\epsilon_t}^{-1} \bm F$ and $\bm \ell_{\beta_t} = \bm F'\bm \Sigma_{\epsilon_t}^{-1} \bm Y_t$.
\end{lemma}
The most useful case of Lemma \ref{ftsLikeGen2} is when $\bm \Sigma_{\epsilon_t}$ is diagonal, so that the error covariance function is $C_{\epsilon_t}(\bm \tau, \bm u) = \mbox{Cov}(\epsilon_t(\bm \tau), \epsilon_t(\bm u)) = \mathbb{I}\{\bm \tau = \bm u\} V_{\epsilon_t}(\bm \tau)$ and $V_{\epsilon_t}(\cdot)$ is the variance function. In this case, computing the inverse $\bm \Sigma_{\epsilon_t}^{-1}$ is efficient, and the projection step to obtain $\bm{\tilde Y_{t}}$ only requires the inverse of a $K\times K$ matrix, $\bm Q_{\beta_t}$. Furthermore, if $V_{\epsilon_t}(\cdot) = V_{\epsilon}(\cdot)$ is non-dynamic, then computing $\bm Q_{\beta_t}^{-1} = \bm Q_{\beta}^{-1}$ is a one-time cost per MCMC iteration. 


\subsection*{MCMC Diagnostics}\label{extra:diag}
We include MCMC diagnostics for the fertility application (Section \ref{fertility}). Traceplots for $\hat Y_t(\bm\tau) = \sum_{k=1}^K f_k(\bm\tau) \beta_{k,t}$ and $\tilde\alpha_j(\bm\tau) = \sum_{k=1}^K f_k(\bm\tau) \alpha_{j,k}$ are in Figures \ref{fig:trace-yhat-fert} and \ref{fig:trace-fert}, respectively. 
These traceplots indicate good mixing and suggest convergence.

\begin{figure}[h]
\begin{center}
\includegraphics[width=1\textwidth]{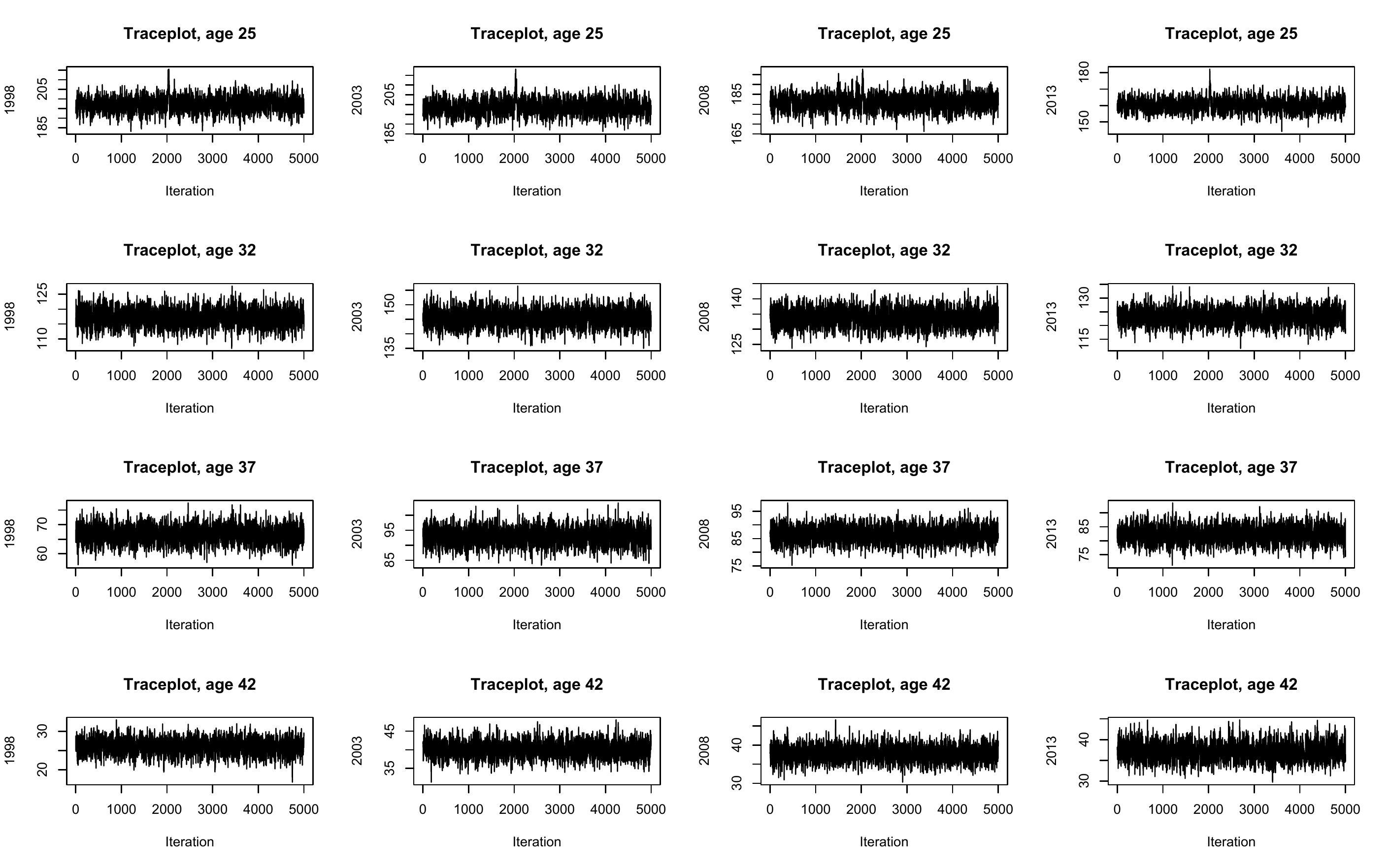}
\caption{Traceplots for the model-imputed ASFR curves $\hat Y_t(\bm\tau) = \sum_{k=1}^K f_k(\bm\tau) \beta_{k,t}$ at various ages $\bm\tau$ for various years $t$ in the fertility application. The traceplots indicate good mixing.
\label{fig:trace-yhat-fert}}
\end{center}
\end{figure}

\begin{figure}[h]
\begin{center}
\includegraphics[width=1\textwidth]{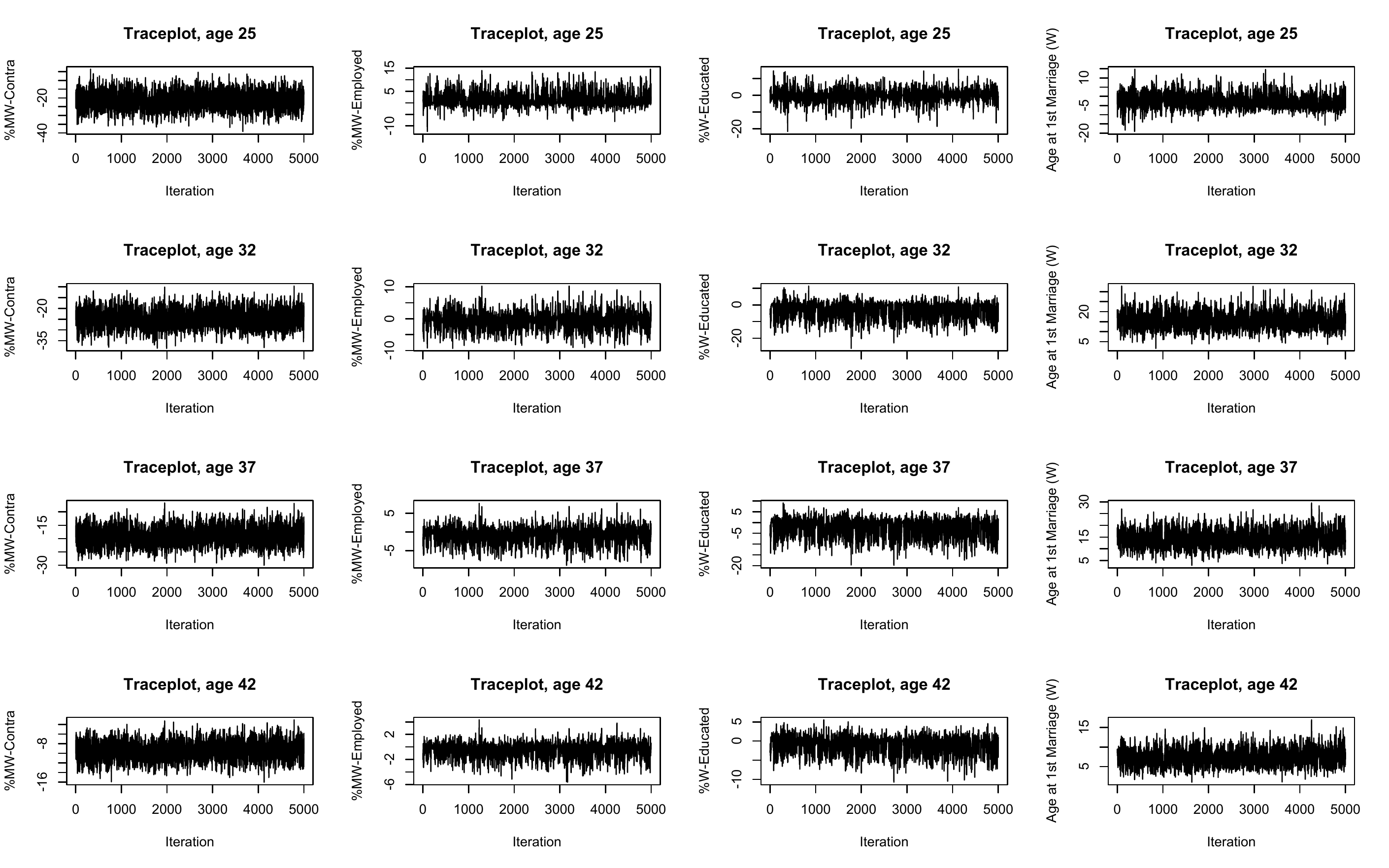}
\caption{Traceplots for the static regression functions $\tilde\alpha_j(\bm\tau) = \sum_{k=1}^K f_k(\bm\tau) \alpha_{j,k}$ at various ages $\bm\tau$ for the predictors in the fertility application. The traceplots indicate acceptable mixing.
\label{fig:trace-fert}}
\end{center}
\end{figure}

\subsection*{Additional Application Details}\label{extra:app}
Figure \ref{fig:sv} plots the observation error standard deviation, $\sigma_{\epsilon_t}$, for the yield curve application. To incorporate volatility clustering, we include a stochastic volatility model for $\sigma_{\epsilon_t}^2$, following \cite{kastner2014ancillarity}. There is strong evidence that the observation error standard deviation is time-varying. Importantly, the proposed DFOSR model framework can incorporate the stochastic volatility model with minimal modifications.

\begin{figure}[h]
\begin{center}
\includegraphics[width=1\textwidth]{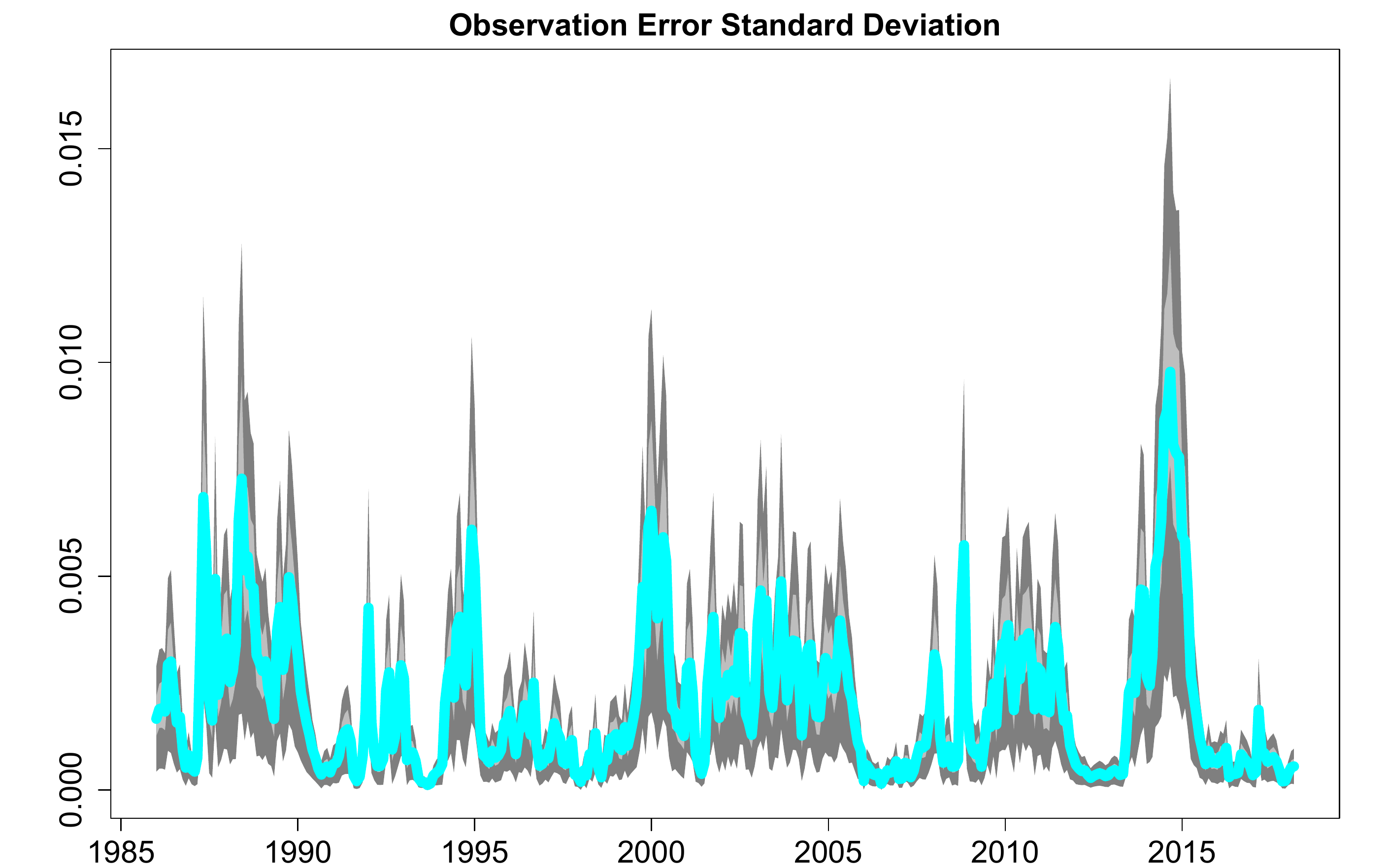}
\caption{Observation error standard deviation, $\sigma_{\epsilon_t}$, for the yield curve data. The solid line is the posterior mean, the light gray bands are 95\% pointwise credible intervals, and the dark gray bands are 95\% simultaneous credible bands. There is strong evidence that the observation error standard deviation is time-varying.
\label{fig:sv}}
\end{center}
\end{figure}

\end{document}